\documentclass[twoside,11pt]{article}

\usepackage{amsmath}
\usepackage{graphicx}
\usepackage{enumerate}
\usepackage{natbib}
\usepackage{url} 

\usepackage{paralist, amsmath, amsthm, amssymb, color, graphicx}

\usepackage{dsfont}
\newcommand{\mscr}{\mathcal}
\newcommand{\twid}[1]{\widetilde{#1}}
\newcommand{\ZZ}{\mathbb{Z}}
\usepackage{thm-restate}
\usepackage[export]{adjustbox}

\usepackage{algorithm,algorithmic}

\usepackage{cancel}
\usepackage{verbatim}
\usepackage{mathtools}
\theoremstyle{plain}
\newtheorem{thm}{Theorem}
\newtheorem{lem}[thm]{Lemma}

\theoremstyle{definition}
\newtheorem{defn}[thm]{Definition} 
\newtheorem{definition}[thm]{Definition} 
 
\newtheorem{example}[thm]{Example}
\newtheorem{remark}[thm]{Remark}

\newcommand{\iid}{\overset{\text{iid}}{\sim}}

\newcommand{\defeq}{\vcentcolon=}

\newcommand{\RR}{\mathbb{R}}

\newcommand{\EE}{\mathbb{E}}

\DeclareMathOperator{\var}{Var}
\DeclareMathOperator{\cov}{Cov}

\newcommand{\ul}{\underline}
\newcommand{\ol}{\overline}

\DeclareMathOperator{\vect}{vec}

\newcommand\blfootnote[1]{%
  \begingroup
  \renewcommand\thefootnote{}\footnote{#1}%
  \addtocounter{footnote}{-1}%
  \endgroup
}

\usepackage{jmlr2e}

\allowdisplaybreaks

\usepackage{lastpage}
\jmlrheading{xx}{xxxx}{1-\pageref{LastPage}}{x/xx; Revised x/xx}{x/xx}{xx-xxxx}{Jordan Awan, Adam Edwards, Paul Bartholomew, and Andrew Sillers}

\ShortHeadings{BLUE from Privatized Contingency Tables}{Awan, Edwards, Bartholomew, and Sillers}
\firstpageno{1}

\begin{document}

\title{Best Linear Unbiased Estimate\\  from Privatized Contingency Tables}

\author{\name Jordan Awan
\email jawan@purdue.edu \\
       \addr Department of Statistics\\
       Purdue University\\
       West Lafayette, IN 47907, USA
       \AND
       \name Adam Edwards \email amedwards@mitre.org \\
       \addr The MITRE Corporation\\
       McLean, VA 22102, USA
       \AND
       \name Paul Bartholomew \email pbartholomew@mitre.org \\
       \addr The MITRE Corporation\\
       McLean, VA 22102, USA
       \AND
       \name Andrew Sillers \email asillers@mitre.org\\
       \addr The MITRE Corporation\\
       McLean, VA 22102, USA}
       \blfootnote{Jordan Awan and Adam Edwards are co-first authors and contributed equally to this work.} 
       \blfootnote{Approved for Public Release; Distribution Unlimited. Public Release Case Number 24-2176.}

\editor{}
\thispagestyle{empty}
\maketitle

\begin{abstract}
In differential privacy (DP) mechanisms, it can be beneficial to release ``redundant'' outputs,  where some quantities can be estimated in multiple ways by  combining different privatized values. Indeed, the DP 2020 Decennial Census products published by the U.S. Census Bureau  consist of such redundant noisy counts. When redundancy is present, the DP output can be improved by enforcing self-consistency (i.e.,  estimators obtained using different noisy counts result in the same value), and we show that the minimum variance processing is a linear projection. 
  However, standard projection algorithms require excessive computation and memory, making them impractical for large-scale applications such as the Decennial Census. 
  We propose the Scalable Efficient Algorithm for Best Linear Unbiased Estimate (SEA BLUE), based on a two-step process of aggregation and differencing that 
  1) enforces self-consistency through a linear and unbiased procedure, 
  2) is computationally and memory efficient, 
  3) achieves the minimum variance solution under certain structural assumptions, and 
  4) is empirically shown to be robust to violations of these structural assumptions. 
   We propose three methods of calculating confidence intervals from our estimates, under various assumptions. Finally, we apply SEA BLUE to two 2010 Census demonstration products, illustrating its scalability and validity. 
\end{abstract}


\section{Introduction}
\label{sec:intro}
Differential privacy (DP) is a formal privacy framework, proposed by \citet{dwork2006calibrating}, which requires introducing randomness into a statistical procedure to provide provable privacy protection. The U.S. Census Bureau has employed differential privacy techniques to protect the 2020 Decennial Census \citep{abowd2019census}, which results in uncertainty about the true counts. The differential privacy mechanism used by the U.S. Census Bureau in 2020 produces millions of noisy counts which are ``redundant'' in that the same values are counted at various levels of geography (e.g., block, county, state) as well as different levels of detail (e.g., different $k$-way contingency tables). However, because independent noise is added to each count, these values are not \emph{self-consistent}, meaning that detailed noisy counts do not add up to marginal noisy counts (see Figure \ref{fig:tables} for an example). The lack of self-consistency is important because different combinations of the noisy counts can give different estimates of the same quantity, which can be a problem when a single ``official'' estimate is needed. Furthermore, the existence of multiple independent estimates implies that each individual estimate is sub-optimal, whereas the collection of estimates could be aggregated to produce a single refined estimate.

To process the noisy counts, the U.S. Census Bureau applied the \emph{TopDown Algorithm}, which enforces self-consistency, non-negativity, and integer-valued constraints. While their method substantially improves the average error, it also introduces bias in the DP estimates, and as a non-linear procedure it is difficult to quantify the uncertainty in the resulting estimates \citep{santos2020differential,kenny2021use,winkler2021differential}. As an alternative, we propose a linear procedure which enforces self-consistency but allows negative and non-integer-valued estimates. 
This approach has the following  benefits: 1) the procedure results in unbiased estimates, and 2) because it is linear, we can accurately estimate the error and distribution of the final estimates to produce confidence intervals for the true counts. 

We first show that the best linear unbiased estimator (BLUE) for the true counts assuming only self-consistency is a linear projection which can be expressed in terms of the self-consistency constraints. This result follows from the Gauss-Markov theorem and agrees with similar prior results in the DP literature (e.g., \citealp{gao2022subspace}). However, implementing this projection by standard techniques requires inverting a large matrix, resulting in a high computational and memory cost. This makes this approach inapplicable for large applications, such as the Demographic and Housing Characteristics (DHC) Census product. 

To address this limitation, we propose a novel method that we call the Scalable Efficient Algorithm for the Best Linear Unbiased Estimator (SEA BLUE), which consists of a two step process of aggregation and differencing. This procedure is inspired by the up and down passes of \citet{hay2010boosting,honaker2015efficient,kifer2021}, which process noisy values in a tree structure. However, in our setting detailed counts are constrained to be self-consistent with multiple marginal counts, so we do not have a tree, but a much more complex structure. This makes the identification of the analogous process non-trivial. Our procedure:
\begin{enumerate}
  \item   enforces self-consistency through a linear and unbiased procedure, 
  \item   is computationally and memory efficient, 
  \item   produces the best linear unbiased estimate under certain structural assumptions, and 
  \item   is empirically demonstrated to be robust to mild violations of the assumptions. 
\end{enumerate}
 
 Using the linear nature of our procedure, we propose three methods of producing confidence intervals for the original counts. Two of these methods calculate the resulting standard error of the estimates either exactly or via Monte Carlo and produce confidence intervals using a normality assumption. The third method is a distribution-free Monte Carlo method, which gives valid confidence intervals which are valid for all noise distributions. 
 
 We illustrate our method in simulation studies, comparing mean-squared error (MSE), confidence interval coverage/width, and computational time/memory. We also apply our methodology to two 2010 Census demonstration products, Redistricting Data (Public Law 94-171, which we abbreviate as PL94) and Demographic and Housing
Characteristics (DHC), illustrating the scalability and validity of our methods. Proofs and technical details are found in the appendix.

While our method is motivated by the 2020 Census products, it can be applied in other settings where independent noisy counts are observed at varying levels of detail. 
\subsection{Related  Work}
The first processing of redundant DP measurements was due to \citet{hay2010boosting}, which leveraged a binary tree structure of the different true values. Their method used a two stage up and down pass to sequentially update the estimates and finally resulted in minimum variance estimates that are self-consistent. A similar approach has been considered by \citet{honaker2015efficient} and \citet{kifer2021}, and this method has been applied to derive DP confidence intervals for the median of real-valued data \citep{drechsler2022nonparametric}. Recently \citet{cumings2024full} extended this tree-based approach to process the  Census noisy measurement files across geographies, which forms a tree. In contrast, SEA BLUE is designed to optimize estimates within geographies using the structure of different marginal counts, which does not form a tree.

\citet{mccartan2023making} and \citet{kenny2024evaluating} proposed a minimum variance weighting of estimates which can be summed up from other tables, which is the same as our collection step, which forms the first step of our SEA BLUE procedure. \citet{kenny2024evaluating} compare these unbiased weighted estimates to the TopDown estimates as well as the estimates from swapping, the privacy protection algorithm used in 2010. While the collection step estimate is optimal ``from below,'' as we formally prove in Section \ref{s:collection}, it is not the best linear unbiased estimate as it does not use all information available in the tables. We show in Section \ref{s:downstep} that our down pass is needed after the collection step to obtain the BLUE. 

 \section{Background and Notation}
 In this section, we review some basic results in linear algebra and differential privacy, and set the notation for the paper -- especially the notation used for contingency tables. 
\subsection{Linear Algebra}
In $\RR^n$ the standard inner product is $\langle x,y\rangle=x^\top y$. To account for covariance structure, other inner products also exist: for a square and positive definite matrix $\Sigma$, an inner product is $\langle x,y\rangle_{\Sigma} = x^\top \Sigma y$, which has corresponding norm  $\lVert x\rVert_\Sigma=\sqrt{\langle x,x\rangle_\Sigma}$.

For the inner product $\langle\cdot,\cdot\rangle_\Sigma$ and a tuple of linearly independent vectors in $\RR^n$, $V = (v_1,\ldots, v_k)$, the projection operator onto $S_V=\mathrm{span}(V)$ is $\mathrm{Proj}^\Sigma_{S_V} = V(V^\top \Sigma V)^{-1} V^\top\Sigma$. The projection is the solution to the following ``least squares'' problem: Let $y\in \RR^n$; then $\mathrm{Proj}_{S_V}^\Sigma y = \arg\min_{x\in S_V} \lVert x-y\rVert_\Sigma$. Furthermore, $\mathrm{Proj}^{\Sigma^{-1}}_{S_V} = I-\mathrm{Proj}^{\Sigma}_{S_V^\perp}$, where $S_V^\perp$ is the orthogonal complement of $S_V$.  


 \subsection{Contingency Table Notation}\label{s:tableNotation}
 In this section we introduce some specialized notation that is necessary to precisely describe our proposed methodology. A summary of the notation is found in Table \ref{tab:notation} in Appendix \ref{s:notation}.
 
\begin{figure}[t]
\centering
     \includegraphics[width=.27\linewidth,valign=t]{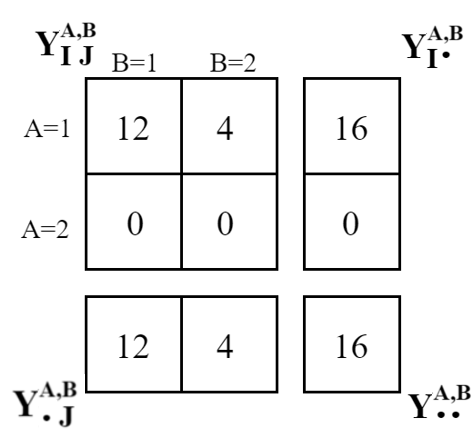}
          \includegraphics[width=.315\linewidth,valign=t]{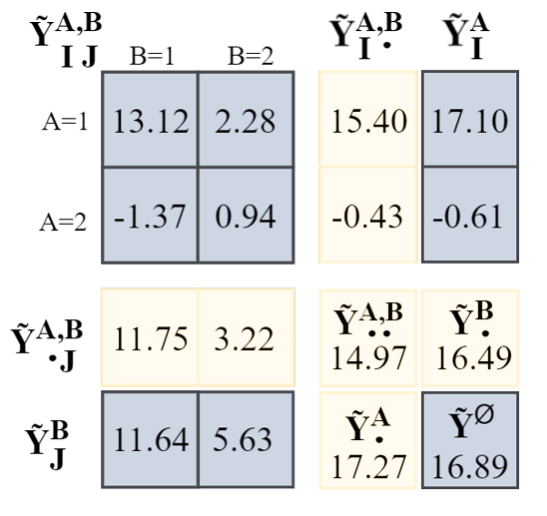}
     \includegraphics[width=.27\linewidth,valign=t]{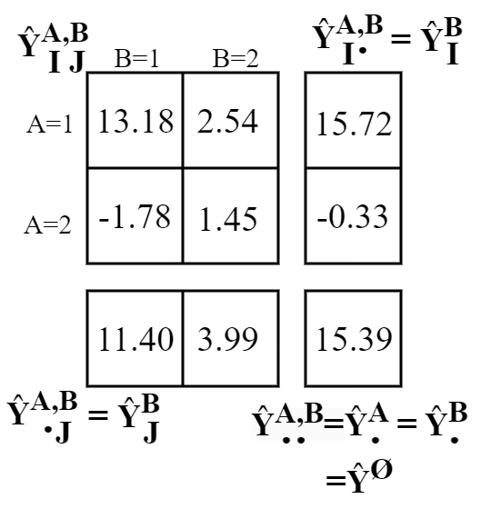}
     \caption{Left: true sensitive counts. Note that they are self-consistent with the margins. Middle: Dark blue values are raw noisy measurements and light yellow values are computed margins. Note that the values are not self-consistent. Right: Final processed estimates, which are self-consistent with their margins. }
     \label{fig:tables}
 \end{figure}
We generally use the variables  $Y$ for true (unobserved) counts, $\tilde Y$ for the observed noisy counts, $\acute Y$ for intermediate estimates of $Y$ from the collection step of the algorithm, and $\hat Y$ for final estimators; see Example \ref{ex:tables} and Figure \ref{fig:tables} for an example. Given a finite set of variables $\mscr A = \{A,B,C,\ldots\}$ with an ordering (such as alphabetical), and $S\subset \mscr A$ a subset of these variables, we write $\twid Y^S$ as the table of counts marginalized over the variables not in $S$ to which noise has been applied. For a variable $A\in \mscr A$, $|A|$ denotes the number of levels that $A$ can take on. If $S=\{A_1,\ldots, A_k\}\subset \mscr A$, then  $V^S=\{1,2,\ldots, |A_1|\}\times \cdots\times \{1,2,\ldots, |A_k|\}$ is the set of indices for $Y^S$; if $v\in V^S$ is a vector, then $\twid Y^S_v$ represents the single noisy count in $\twid Y^S$ which corresponds to $A_1=v_1$, $A_2=v_2$, and so on. We also consider $V^S_\bullet=\{\bullet,1,2,\ldots |A_1|\}\times \cdots\times \{\bullet,1,2,\ldots, |A_k|\}$, which allows for a ``$\bullet$'' in some of the entries. If $v\in V^S_\bullet$, the entries with a ``$\bullet$''  indicate that we are summing over all of the values for the corresponding variable. We use the same notation for $Y$, $\acute Y$, and $\hat Y$ as well. For specific examples, we simplify the notation by dropping braces on the sets $S$ and vectors $v$ (e.g., for $S=\{A,B\}$ and $v=(i,j)$, we  write $\twid Y_{i,j}^{A,B}$ in place of  $\twid Y_{(i,j)}^{\{A,B\}}$.

We use the term \emph{table} to refer to both subsets of variables $S\subset \mscr A$ as well as sets of counts of the form $\twid Y^S$ (similarly for $\acute Y^S$ and $\hat Y^S$). It should be clear from context which type of object is being referred to when the term ``table'' is used. When an index is given, we call objects like $\twid Y^S_v$ and $\hat Y^S_v$ either \emph{counts} or \emph{estimates} of $Y^S_v$. 

\begin{example}\label{ex:tables}
Suppose that we have two binary variables  $A$ and $B$ (e.g., Hispanic and Voting Age). Let $Y_{i,j}$ denote the true counts for $A=i$ and $B=j$. We observe noisy counts $\twid Y^{A,B}_{i,j}=Y_{i,j}+\text{noise}$, $\twid Y^{A}_{i}=Y_{i,\bullet}+\text{noise}$, $\twid Y^B_j=Y_{\bullet,j}+\text{noise}$ and $\twid Y^\emptyset = Y_{\bullet,\bullet}+\text{noise}$. Note that while $Y^A_{i}=Y_{i,\bullet}$, this is generally no longer the case for the values in $\twid Y$: $\twid Y^A_{i}\neq \twid Y^{A,B}_{i,\bullet}$. However, our final estimates $\hat Y$ will be ``self-consistent'' in this sense: e.g.,   $\hat Y^A_i=\hat Y^{A,B}_{i,\bullet}$. See Figure \ref{fig:tables} for an illustration.
\end{example}

Some additional notation that will be important to describe and analyze our procedure is as follows:  If $S\subset \mscr A$, we write $\ul S\subset S$ and $\ol S \supset S$ as notation to remind us that $\ul S$ is a subset of $S$ and $\ol S$ is a superset of $S$. If $v\in V^S$ and $\ol S\supset S$, then we define $v(\ol S)\in V^{\ol S}_\bullet$ as the vector that agrees with $v$ on the entries corresponding to $S$ and has the symbol ``$\bullet$'' for the entries corresponding to $\ol S\setminus S$; the effect of this notation is to marginalize out the variables $\ol S\setminus S$ and use $v$ to index the variables from $S$ in the marginalized table. Similarly, for $\ul S\subset S$, we define $v[\ul S]\in V^{\ul S}$ as the vector which extracts the entries of $v$ which correspond to $\ul S$.  
Combining these two notations, we can also write $v[\ul S](S)\in V^S_\bullet$ for $\ul S\subset S$ and $v\in V^S$, and this vector agrees with $v$ on $\ul S$ and has ``$\bullet$'' on $S\setminus \ul S$, which has the effect of marginalizing out $S\setminus \ul S$ and using the remaining values in $v$ to index the marginal table on $\ul S$.

 \subsection{Differential Privacy}\label{s:DP}
 Differential privacy (DP) is a probabilistic framework which formally quantifies the amount of privacy protection offered by a mechanism (randomized algorithm) \citep{dwork2006calibrating}.  Intuitively, a mechanism satisfies DP if, when it is run on two databases differing in one person's contribution, the distribution of outputs is similar.  There are many different notions of differential privacy, such as pure-DP, approximate-DP,  Gaussian-DP \citep{dong2022gaussian}, and zero-concentrated DP \citep{bun2016concentrated}, which differ in how the ``similarity'' between distributions is measured. In zero-concentrated DP, the version of DP employed by the U.S. Census Bureau \citep{Abowd20222020}, the ``similarity'' is measured in terms of divergences on probability measures. 

 \begin{definition}
[Zero-Concentrated DP: \cite{bun2016concentrated}] Let $\rho\geq 0$. A mechanism $M:\mscr X\rightarrow \mscr Y$ satisfies $\rho$-zero concentrated differential privacy ($\rho$-zCDP) if  for any two databases $X$ and $X'$ differing in one entry, we have 
$D_{\alpha}(M(X)||M(X'))\leq \rho \alpha, $ 
for all $\alpha\in (1,\infty)$, where $D_\alpha(\cdot ||\cdot)$ is the order $\alpha$ R\'enyi divergence on probability measures:
\[D_\alpha(P||Q) = \frac{1}{\alpha-1} \log \left( \EE_{X\sim P} \left[ \frac{p(X)}{q(X)}\right]^{\alpha-1}\right),\] where $p$ and $q$ are the probability mass/density functions for $P$ and $Q$, respectively. 
 \end{definition}

In $\rho$-zCDP, smaller values of $\rho$ ensure that the distributions of $M(X)$ and $M(X')$ are more similar, offering a stronger privacy guarantee.  

The simplest and most widely used method of achieving differential privacy is to add independent noise to the entries of a statistic. The noise must be scaled to the \emph{sensitivity} of the statistic. A statistic $T:\mscr X\rightarrow \mathbb{R}^k$ has $\ell_2$-sensitivity $\Delta$ if 
$\Delta \geq \sup_{X,X'} \lVert T(X)-T(X')\rVert_2,$ 
where the supremum is over all $X$ and $X'$ which differ in one entry. If $T$ has $\ell_2$-sensitivity $\Delta$, then $T+N$  satisfies $\rho$-zCDP, where $N\sim \mscr N(0,\Delta^2/(2\rho)I)$. If $T$ takes values in $\ZZ^n$ and has $\ell_2$-sensitivity $\Delta$, then $T_i+N_i$, where $N_i\iid N_{\ZZ}(0,\Delta^2/(2\rho))$ satisfies $\rho$-zCDP. Here, $\mscr N_{\ZZ}(0,\sigma^2)$ is the \emph{discrete Gaussian} distribution with probability mass function $P(X=t)\propto \exp(-t^2/(2\sigma^2))$ for $t\in \ZZ$ \citep{canonne2020discrete}. The U.S. Census Bureau used discrete Gaussian noise to produce noisy measurements of the 2020 Decennial Census tabulations,  which satisfy zCDP \citep{Abowd20222020}.  The methodology developed in this paper can be applied post-process a DP output from any DP framework as long as additive noise is used.

Differential privacy has several important properties including composition, group privacy, and immunity to post-processing \citep{dwork2014algorithmic,bun2016concentrated}. We review immunity to post-processing as it is central to the problem tackled in this paper. 


{\bf Immunity to Post-Processing: } If $M:\mscr X\rightarrow \mscr Y$ satisfies $\rho$-zCDP, and $g:\mscr Y\rightarrow \mscr Z$ is a (possibly randomized) function that does not depend on the sensitive data, then $g\circ M:\mscr X\rightarrow \mscr Z$ satisfies $\rho$-zCDP as well. Immunity to post-processing is important because one can first design a privacy mechanism which produces summary statistics, and then process these without compromising the privacy guarantee. Furthermore, by the transparency property of DP, the mechanism $M$ itself can be publicly known and incorporated into the post-processing. For example, the ``noisy measurements'' (which we denote by $\twid Y$) published by the U.S. Census Bureau satisfy zCDP; by post-processing, the procedure we propose can give refined estimates $\hat Y$ without compromising the privacy guarantee\footnote{Technically the Census noisy measurement files do not strictly satisfy differential privacy because they contain \emph{invariants}, which are values from the data reported without noise. See \citet{gao2022subspace} and \citet{cho2024formal} for investigations on what the technical guarantee is for these Census products}.

\section{Problem Setup and Theoretical BLUE}
For the problem setup, we assume that we observe noisy counts $\twid Y^S_v=Y^S_{v}+N^S_v$ for $S\in \mscr O\subset \{A,B,C,\ldots\}$ and $v\in V^S$, where $\mscr O$ denotes the set of \emph{observed tables}, $Y$ is the original contingency table of sensitive counts, and $N^S_v$ are the noise random variables, which have a known distribution. Our goal is to produce refined estimates $\hat Y_v^S$ for $S\in \mscr D$, where $\mscr D$ is the set of \emph{desired tables}, such that the $\hat Y_v^S$ are \emph{self-consistent}, meaning that we have $\hat Y^S_{v(S)}=\hat Y^{\ul S}_{v}$ for all $S, \ul S\in \mscr D$, such that $\ul S\subset S$ and for all $v\in V^{\ul S}$. In other words, marginalizing out the variables of $S\setminus \ul S$ in $\hat Y^S$ results in the same values as in $\hat Y^{\ul S}$. We assume that the tables in $\mscr D$ are nested: if $S\in \mscr D$ and $\ul S\subset S$, then $\ul S\in \mscr D$, but $\mscr O$ need not have this structure. Furthermore, we require our final estimates to be linear and unbiased, with the smallest variance possible. 

 For the remainder of the paper, the true counts $Y$ will be viewed as fixed, unknown quantities. Thus, statements about the distribution of $\twid Y$, or other random variables, refer only to the randomness due to the noise mechanism.

We consider a few additional structural assumptions on the problem:
\begin{itemize}
    \item [(A1)][Known noise] The noise random variables $N^S_v$ are mean zero with known covariance $\Sigma=\cov(N)$. Furthermore, $\Sigma$ is positive-definite and consists of finite entries.
    \item [(A2)] [Uncorrelated noise] The $N^S_v$ are uncorrelated, implying that $\Sigma$ is a diagonal matrix.
    \item [(A3)] [Equal variance within tables] The noise variance is constant within a table: $\var(\twid Y_v^S)=\var(\twid Y_w^S)$ for all $v,w\in V^S$ and $S\in \mscr O$.
\end{itemize}

Note that (A1)-(A3) are all verifiable assumptions on the problem setup. 
Furthermore, by the transparency property of DP, the privacy mechanism used can be published along with the privatized data $\twid Y$ without weakening the theoretical privacy guarantee, to allow outside users to check these assumptions as well.  Assumptions (A1) and (A2) were previously used in the literature \citep{hay2010boosting,honaker2015efficient,kifer2021}, while (A3) is new to this work and is needed to limit the search of all possible estimators.

\begin{example}[Motivating Census Data Products]
   The key motivating data sets for this paper are noisy measurement files for the 2020 Decennial Census data products, PL94 and DHC. Restricted to a single geography (e.g., only county-level counts for a particular county), the PL94 product satisfies (A1)-(A3); DHC does not exactly fit our framework as Age variable is binned to different values in different tables. In our application to DHC in Section \ref{s:simulations}, we use a subset of the noisy counts, which do satisfy (A1)-(A3). More details on these Census products is provided in Section \ref{s:products}.
\end{example}

\begin{remark}[Invariants]\label{rem:invariant}
    Some of the Census tables also have \emph{invariant} counts, which are values observed without noise. While such counts do not explicitly fit into the assumptions (A1)-(A3), with the present version of our method, our procedure can be applied by setting the variance to be  arbitrarily small  for these counts.  We suspect that a modification to our procedure can formally incorporate these invariants as part of the constraints without such an approximation, but leave this as future work.
\end{remark}

A general solution can be provided to this problem, assuming only property (A1), using the theory of linear models. We show that the best linear unbiased estimator, can be expressed in terms of linear projections, and  we give a formula for the projection operator in terms of the constraint matrix. The proof of the result is based on the theory of linearly constrained linear regression. 

\begin{restatable}{thm}{thmgeneral}\label{thm:general}
Let $X\in \RR^n$ be a random vector with mean $\mu$ and covariance $\Sigma$, 
where it is known that $\mu \in S_b = \{x\in \RR^n\mid A x=b\}$. Let $v$ be such that $b=Av$ and call $S_0=\{x\in \RR^n\mid Ax=0\}$. Call $\hat \mu = \mathrm{Proj}_{S_0}^{\Sigma^{-1}}(X)+(\mathrm{Proj}_{S_0^\perp}^{\Sigma})^\top v$. Then, 
\begin{itemize}
    \item $\hat \mu$ is the BLUE for $\mu$. Furthermore, $\mathrm{Cov}(\hat\mu)=\mathrm{Proj}_{S_0}^{\Sigma^{-1}}\Sigma (\mathrm{Proj}_{S_0}^{\Sigma^{-1}})^\top$. 
    \item If $X\sim \mscr N(\mu,\Sigma)$, then $\hat\mu$ is the uniformly minimum variance unbiased estimator for $\mu$; furthermore, $\hat\mu \sim \mscr N(\mu, \mathrm{Proj}_{S_0}^{\Sigma^{-1}}\Sigma (\mathrm{Proj}_{S_0}^{\Sigma^{-1}})^\top)$.
\end{itemize}
 If the rows of $A$ are linearly independent, then 
 \[\mathrm{Proj}_{S_0}^{\Sigma^{-1}}=\left(I-\Sigma A^\top (A\Sigma A^\top)^{-1} A\right)\qquad\text{and} \qquad(\mathrm{Proj}_{S_0^\perp}^{\Sigma})^\top=\Sigma A^\top (A\Sigma A^\top)^{-1} A.\]
\end{restatable}

While Theorem \ref{thm:general} derives the BLUE, when  $n$ becomes large, computing the projection becomes expensive, especially in terms of memory which scales as $O(pn)$, where $p$ is the dimension of the rowspace or nullspace of $A$ -- whichever is smaller. Furthermore, the runtime for this estimator using matrix operations is approximately $O(n^{2.373})$ \citep{alman2021refined}; see Section \ref{s:complexity} for a discussion of computational complexity.

\begin{example}
 Applying Theorem \ref{thm:general} to our setting, we have $X=\tilde Y$, $\mu=Y$, and $\hat\mu=\hat Y$. In the constraints, $A$ is the matrix with entries in $\{-1,0,1\}$, which encodes the self-consistency constraints, and $b=0$. Note that the projections are orthogonal in the case that $\Sigma=cI$; in general, assumptions (A1)-(A3) do not ensure that the projections are orthogonal.
\end{example}

\begin{example}
    [Running Toy Example]\label{ex:toy}
    We introduce a toy example that we will use to illustrate Theorem \ref{thm:general}; later we will revisit this example to illustrate the SEA BLUE algorithm. Suppose that the only variable is $B$, which has 3 levels. The sensitive values are $Y^B=(5,10,15)^\top$, with $Y^\emptyset=30$. After adding independent discrete Gaussian noise with mean zero and variance 1 to each value, we observe the noisy counts $\twid Y^B=(6,9,17)^\top$ and $\twid Y^\emptyset=29$. Note that the noisy values are not self-consistent as $\twid Y^B_\bullet=32 \neq 29=\twid Y^\emptyset$. In the notation of Theorem \ref{thm:general}, we have $X=\binom{\twid Y^B}{\twid Y^\emptyset}$ with covariance $\Sigma=I$, and constraint $AX=0$ where $A=(1,1,1,-1)$. We calculate that 
    \[\mathrm{Proj}_{S_0}^I=(I-A^\top(AA^\top)^{-1}A)=\frac 14\left(\begin{array}{rrrr}
    3&-1&-1&1\\
    -1&3&-1&1\\
    -1&-1&3&1\\
    1&1&1&3\end{array}\right).\]
    Thus, by Theorem \ref{thm:general}, we have that $\binom{\hat Y^B}{\hat Y^\emptyset}=\mathrm{Proj}_{S_0}^I\binom{\twid Y^B}{\twid Y^\emptyset}=(5.25,8.25,16.25,29.75)^\top,$ which are self-consistent and BLUE.
\end{example}

In the following section, we propose an efficient implementation of BLUE under the additional assumptions of uncorrelated noise (A2) and equal variance within tables (A3). 

\section{Scalable Efficient Algorithm for BLUE}
In this section, we propose a novel two-step estimator that implements BLUE under assumptions (A1)-(A3) and gives a well-motivated heuristic when (A2) and (A3) do not hold. We call our method the Scalable Efficient Algorithm for the Best Linear Unbiased Estimator (SEA BLUE). The two steps of SEA BLUE are a collection step, which aggregates estimators of lower margins from more detailed tables, followed by a down pass, which projects the intermediate estimators resulting from the first step to be self-consistent with the previously calculated margins above. See Figure \ref{fig:Passes} for an illustration.

In Section \ref{s:collection}, we describe the collection step of SEA BLUE and prove in Theorem \ref{thm:total} that it results in the BLUE based on all estimators ``from below;'' in the case of the estimate of the total count, this is simply the BLUE. We also present a pseudo-code implementation of the collection step, which optimizes performance. In Section \ref{s:downstep} we describe the down pass and prove that, when applied to the output of the collection step, it results in the BLUE for every table. We also develop an efficient implementation of the down pass by optimizing the particular projection operations that are used in this step, and give a pseudo-code implementation. In Section \ref{s:complexity}, we compare both the runtime and memory complexity of SEA BLUE to the matrix projection, demonstrating the increased efficiency of SEA BLUE. 

 \begin{figure}[t]
     \centering
          \includegraphics[width=.35937\linewidth]{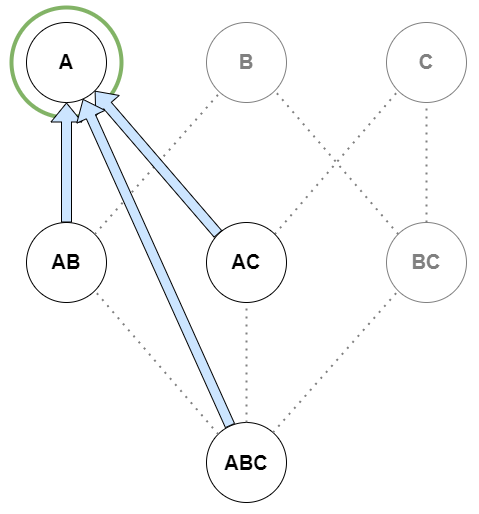}
     \includegraphics[width=.4492125\linewidth]{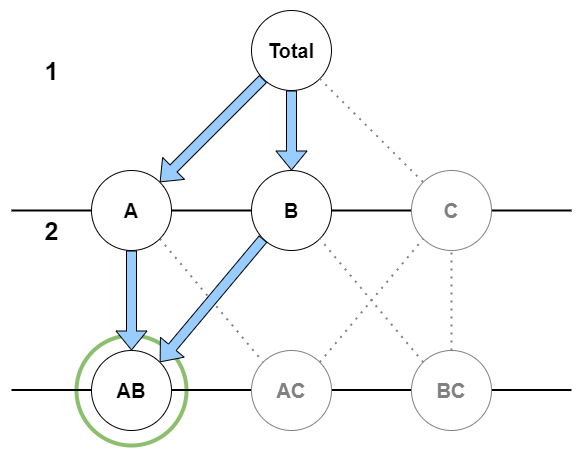}
     \caption{ Left: Illustration of the ``collection step,'' which updates  the estimates in the $A$ table using a weighted average of the $A$ margin calculated from $A$, $AB$, $AC$, and $ABC$ tables from the original noisy counts in $\twid Y$. Right: Illustration of the ``down pass,'' which  updates the $AB$ table to be consistent with the $A$ and $B$ tables via projection; prior to this, the $A$, $B$, and $C$ tables were projected to be consistent with the total. The down pass is applied to the intermediate estimates $\acute Y$ obtained from the collection step.}
     \label{fig:Passes}
 \end{figure}


\subsection{Collection Step}\label{s:collection}
The first step of SEA BLUE is called the \emph{collection step}, which considers estimators that can be  summed up from a single table,  and weights these to give preliminary estimators of each count, which we denote by $\acute Y_v^S$ for a table $S$ and index $v\in V^S$. The collection step was designed with the assumption of uncorrelated noise (A2) in mind, which removes the need to consider the covariance between estimators from different tables. By using inverse-variance weighting at the end of the collection step, the intermediate estimators are optimal ``from below.''  The collection step does not consider any estimators which leverage information between tables that are not hierarchically related, and the down pass will be needed to incorporate this information. 


\begin{defn}
    [Collection Step] Let $\twid Y_v^S = Y_{v}^S + N_v^S$ be noisy counts for $S\in \mscr O$ and $v\in V^S$ satisfying (A1). For any $S\subset \mscr A$ and $v\in V^S$, the \emph{collection step} produces intermediate estimates $\acute{Y}_v^S$ via the following formula:
\begin{equation}\label{eq:collection}
\acute Y_v^S = \frac{\sum\limits_{\ol S\in \mscr O | \ol S\supset S} \left(\var(\twid Y^{\ol S}_{v(\ol S)})\right)^{-1} \twid Y_{v(\ol S)}^{\ol S}}{\sum\limits_{\ol S\in \mscr O | \ol S\supset S} \left(\var(\twid Y^{\ol S}_{v(\ol S)})\right)^{-1}}.
\end{equation}
\end{defn}

Note that because the entries of the noisy counts $\twid Y$ have a known covariance, we can easily calculate $\var(\twid Y^{\ol S}_{v(\ol S)})$ exactly. If (A2) holds, then the variance of the collection step estimates is given by the following Lemma:
\begin{restatable}{lemma}{lemcollectionVar}\label{lem:collectionVar}
If (A1) and (A2) hold, then the variance of the entries of the intermediate estimate $\acute Y_v^S$, denoted by $ (\acute\sigma^2)^S_{v}$, is given by
\begin{equation}\label{eq:collectionVar} (\acute\sigma^2)^S_v \defeq \var(\acute Y_v^S)=\left(\sum\limits_{\ol S\in \mscr O | \ol S\supset S} \left(\var(\twid Y^{\ol S}_{v(\ol S)})\right)^{-1}\right)^{-1}.\end{equation}
\end{restatable}

Under assumptions (A1) and (A2), the collection step produces an optimal estimator, using the estimates that are ``from below'' and ``within-table.''  
Under the assumption of equal variance within tables (A3),  the collection step is simply the optimal estimator ``from below.'' In particular, the collection step estimate of the total $Y^\emptyset$ is the  BLUE.

 \begin{restatable}{thm}{thmtotal}\label{thm:total}
 [Optimality for the Total]
Assume that (A1)-(A3) hold. Then, 
     $\acute Y^\emptyset$ is the BLUE of the total count $Y^\emptyset$ based on the values in $\twid Y^S$ for $S\in \mscr O$, subject to self-consistency. 
 \end{restatable}
 While Theorem \ref{thm:total} is explicitly stated in terms of the total $Y^\emptyset$, it implies that all collection estimates are optimal ``from below,'' by relabeling a count of interest $Y^S_v$ as $Y^\emptyset$ and only considering tables $\ol S\in \mscr O$ which satisfy $\ol S\supset S$.

 \begin{example}[Running Toy Example]\label{ex:toy2}
We apply the collection step to the setting in Example \ref{ex:toy}. Note that assumptions (A1)-(A3) all hold. Since there are no observed tables ``below'' $\twid Y^B$, we set $\acute Y^B=\twid Y^B$. Using inverse-variance weighting, we have 
\[\acute Y^\emptyset=\frac{(1/3)\twid Y^B_{\bullet}+(1)\twid Y^\emptyset}{4/3}=\frac{(1/3)(32)+29}{4/3}=29.75.\]
Note that we still do not have self consistency between $\acute Y^B$ and $\acute Y^\emptyset$.
 \end{example}

 Assumption (A2) is needed for Theorem \ref{thm:total} as the collection step only uses inverse-variance weighting, based on the assumption of uncorrelated noise. If noises are correlated, inverse-covariance weighting would be needed \citep{anderson2008multiple}, which would be far more computationally and memory expensive. Fortunately many DP statistics use uncorrelated noise, such as those produced by the U.S. Census Bureau, and prior works have also assumed this structure \citep{hay2010boosting,honaker2015efficient,kifer2021}. The following example demonstrates the suboptimality of $\acute Y^\emptyset$ when (A2) does not hold.  

 \begin{example}[Necessity of (A2) for Collection Step]\label{ex:A2forCollection}
Suppose that $|A|=|B|=2$, $\mscr O = \{\{A\},\{A,B\}\}$, $\twid Y^{A,B}_{i,j}$ are independent with variance 1, and $\twid Y^A_i=\twid Y^{A,B}_{i,\bullet}+N_i$, for $i=1,2$, where $N_i$ are independent of all other variables with variance $1$. Then the collection step produces $\acute Y^\emptyset = ((1/4)\twid Y^{A,B}_{\bullet,\bullet}+(1/6)\twid Y^A_\bullet)(1/4+1/6)=(3\twid Y^{A,B}_{\bullet,\bullet}+2\twid Y^{A}_{\bullet})/5$, which has variance $\var(\twid Y^{A,B}_{\bullet,\bullet})+\frac{4}{25}(\var(N_1)+\var(N_2))=4+8/25$. On the other hand, inverse-covariance weighting tells us that the BLUE for $Y^\emptyset$ is simply $\twid Y^{A,B}_{\bullet,\bullet}$, which has variance 4. Thus, $\acute Y^\emptyset$ is not BLUE.
 \end{example}

It may not be immediately obvious why (A3) is necessary to ensure the optimality of the collection step. However,  if (A3) does not hold, there may be important between-table estimators that are not considered by the collection step, as demonstrated in the following example. 

\begin{example}[Necessity of (A3) for Collection Step]\label{ex:A3forCollection}
Suppose that $|A|=|B|=2$, $\mscr O = \{\{A\},\{A,B\}\}$, $\var(\twid Y^A_1)=\var(\twid Y^{A,B}_{2,1})=\var(\twid Y^{A,B}_{2,2})=1$ and $\var(\twid Y^A_2)=\var(\twid Y^{A,B}_{1,1})=\var(\twid Y^{A,B}_{1,2})=11$. Then (A3) does not hold, since there are unequal variances within both the $\{A,B\}$ table and the $\{A\}$ table. Using \eqref{eq:collectionVar}, we compute 
\[\var(\acute Y^\emptyset)=\left(1/\var(\twid Y^{A}_\bullet)+1/\var(\twid Y^{A,B}_{\bullet,\bullet})\right)^{-1}=(1/12+1/24)^{-1}=8.\] However, a between-table estimator, $Z=\twid Y^A_1 + \twid Y^{A,B}_{2,1} + \twid Y^{A,B}_{2,2}$ has variance $\var(Z)=3$. So, $\acute Y^\emptyset$ is not the BLUE.    
\end{example}

 \subsubsection{Implementation of Collection Step}

    \begin{algorithm}[!ht]
        \begin{algorithmic}[1]
            \REQUIRE Noisy counts $\twid Y^S$ for each $S\in \mscr O$, and known variance $\var(\twid Y^S_v)$ for each $v\in V^S$.

            \STATE For each $S\in \mscr D$ and $v\in V^S$, set $C^S_v=D^S_v=0$.

        \FORALL{$S\in \mscr O$ (loop through observed tables and create a spanning set of margins)}
\STATE Define $\overline{\mscr D}_S\subset 2^S$ such that $\overline{\mscr D}_S\supset \mscr D\cap 2^S$, $\emptyset \in \overline {\mscr D}_S$ and for any $D\in \mscr D\cap 2^S$, there exists a nested sequence $(R_j)_{j=|D|}^{|S|}$ in $\overline{\mscr D}_S$ such that \begin{enumerate}
    \item $R_j\subset R_k$ for all $|D|\leq j\leq k\leq |S|$, 
    \vspace{-.25cm}
    \item $|R_j|=j$,
    \vspace{-.25cm}
    \item $R_{|D|}=D$.
\end{enumerate} 
 (that is, for any desired table $D$, involving the variables in $S$, we have a sequence of tables $(R_j)$ in $\overline {\mscr D}_S$ that increase one variable at a time from $D$ to $S$; we refer to $\overline {\mscr D}_S$ as a ``spanning set of margins'')
\FOR{$i=0,1,\ldots, |S|$ (number of variables to marginalize)}
\FORALL{$R\in \overline{\mscr D}_S$ such that $|S|=|R|-i$ (pick a set of variables not marginalized out)}
\FORALL{$v\in V^R$ (all counts in the table)}
\IF{$i=0$ (implies that $R=S$)}
\STATE Define $\twid Y^{S\mapsto S}_v = \twid Y^S_v$
\STATE Define $(\twid\sigma^2)^{S\mapsto S}_v = \var(\twid Y^S_v)$
\STATE  (the notation $S\mapsto R$ represents that we are producing estimates for the table $R$ based on margins calculated from table $S$)
\ELSE
\STATE There exists $\ol R\in \overline D_S$ such that $R\subset \ol R\subset S$ and $|\ol R|=|R|+1$
\STATE Define $\twid Y^{S\mapsto R}_v = \twid Y^{S\mapsto \ol R}_{v(\ol R)}$ (sum up values from a previously marginalized table)
\vspace{-.4cm}
\STATE Define $(\twid\sigma^2)^{S\mapsto R}_v = (\twid\sigma^2)^{S\mapsto \ol R}_{v(\ol R)}$ (sum up variances from previous margin)
\STATE  (we use a previously calculated margin to avoid re-summing the same quantities unnecessarily)
\ENDIF
\IF{$R\in \mscr D$ (if the table is one of the desired outputs)}
\STATE Update $C^R_v = C^R_v +  \twid Y^{S\mapsto R}_v/(\twid\sigma^2)^{S\mapsto R}_v$  (inverse-variance-weighted running sum)
\STATE Update $D^R_v = D^R_v + 1/(\twid\sigma^2)^{S\mapsto R}_v$  (running sum of the inverse-variances)
\ENDIF
\ENDFOR
\ENDFOR
\ENDFOR
\ENDFOR
\FORALL{$S\in \mscr D$ and $v\in V^S$}
\STATE Set $\acute Y^S_v = C^S_v/D^S_v$  (by Equation \eqref{eq:collection})
\STATE Set $(\acute \sigma^2)^S_v = 1/D^S_v$  (by Equation \eqref{eq:collectionVar})
\ENDFOR
        \end{algorithmic}
        \caption{Pseudo Code for Collection Step}
                \label{alg:collection}
    \end{algorithm}

The collection step \eqref{eq:collection} can be easily implemented by first iterating over all desired tables $S\in \mscr D$, and then over all observed tables $R\supset S$,  to produce the within-table estimates of $Y^S$. However, this implementation is inefficient as it fails to leverage the ability to reuse previous computations. In Algorithm \ref{alg:collection}, we describe an implementation that calculates \eqref{eq:collection} more efficiently. The idea is to first iterate over the observed tables, which allows us to leverage higher detailed margins when calculating lower detailed margins (e.g.,  first calculate the one-way margins and then calculate the total from the one-way margins). Algorithm \ref{alg:collection} offers the possibility of parallelization by working with different observed tables on different computing nodes and aggregating the summary values $C_v^S$, $D_v^S$ at the end. Note that the collection step can be further optimized by marginalizing out the variables with the largest number of levels first in order to reduce the number of counts that need to be marginalized over for lower level margins.

\subsection{Down Pass}\label{s:downstep}

In the down pass, we begin with the intermediate estimates produced by the collection step and iteratively project the $i$-way margins to be self-consistent with the previously computed $(i-1)$-way margin. The total ($0$-way margin) is simply equal to the end result of the collection step ($\hat Y^\emptyset = \acute Y^\emptyset$).

\begin{defn}[Down Pass]\label{def:downpass}
    Let $\acute Y^S_v$ be the output of the collection step, and assume that $(\acute \sigma^2)_v^S=\var(\acute Y^S_v)$ is given (such as by Lemma \ref{lem:collectionVar} using (A1) and (A2)). Then the \emph{down pass} produces final estimates $\hat Y^S$ that are defined recursively as follows:
    \begin{itemize}
        \item If $S=\emptyset$, then $\hat Y^\emptyset = \acute Y^\emptyset$.
        \item Otherwise, $\hat Y^S$ is defined to be the BLUE for $Y^S$, given the intermediate estimates $\acute Y^S$, subject to the constraint $\hat Y^S_{v(S)} = \hat Y^{\ul S}_v$ for all ${\ul S}\subsetneq S$ and $v\in V^{\ul S}$; that is, when marginalizing out the variables of $S\setminus {\ul S}$ in $\hat Y^S$, we obtain the same values as in $\hat Y^{\ul S}$.
    \end{itemize}
By Theorem \ref{thm:general}, we know that $\hat Y^S$ can be expressed as a projection, given in Algorithm \ref{alg:downpass}.
\end{defn}

From Definition \ref{def:downpass}, we see that the down pass is ``locally optimal'' in the sense that if we only had the intermediate estimate $\acute Y^S$, and optimal estimates $\hat Y^{\ul S}$ are treated as fixed for ${\ul S}\subsetneq S$, then $\hat Y^S$ is the optimal estimator. Under assumptions (A1)-(A3), we can in fact say that $\hat Y^S$ is BLUE compared to all linear estimators based on the noisy counts $\twid Y$.  

\begin{restatable}{thm}{thmOptimal}\label{thm:optimal}[Optimality of SEA BLUE]
Assume (A1)-(A3) hold. Then, $\hat Y^S$ resulting from applying the down pass to the output of the collection step is the best linear unbiased estimator of $Y^S$ based on the values in $\twid Y$, under self-consistency.
\end{restatable}

\subsubsection{Efficient Implementation of the Down Pass}
 
 Next, we derive an efficient implementation of the projections required in the down pass: projecting a $k$-way table to be self-consistent with fixed higher margins. Assuming (A1)-(A3) hold, it follows that all of the elements of $\acute Y^S$ are independent with equal variance, and so the desired projection is an \emph{orthogonal projection}.

First, we require some technical notation. Let $S$ be a set of variables, and let $Z^S$ be a generic table on these variables. An important subspace is
 \begin{equation}\label{eq:S}\mscr S_S=\{Z^S|\text{for all } \ul S\subsetneq S, \text{ and all }v\in V^{\ul S}, \text{ we have }Z^S_{v(S)}=0\},\end{equation}
 which is the set of tables with all margins equal to zero. 

Given $\acute Y^S$, the output of the collection step, and $\hat Y^{\ul S}$ for ${\ul S}\subsetneq S$, the desired margin values, we construct two tables: $\hat Z^S$, a table consistent with $\hat Y^{\ul S}$ which does not use the values in $\acute Y^S$, and $\acute Z^S$, which has the same structure as $\hat Z^S$, but uses the  margins of $\acute Y^S$:

\begin{equation}
\label{eq:Zhat}
    \hat Z^S_v =\sum_{k=1}^{|S|}\sum_{\substack{\{A_1,\ldots, A_k\}\subset S\\ \ul S = S\setminus \{A_1,\ldots, A_k\}}} (-1)^{k+1} \frac{\hat Y^R_{v[\ul S]}}{|A_1|\cdots |A_k|},
\end{equation}

\begin{equation}
\label{eq:Zacute}
    \acute Z^S_v =\sum_{k=1}^{|S|}\sum_{\substack{\{A_1,\ldots, A_k\}\subset S\\ \ul S = S\setminus \{A_1,\ldots, A_k\}}} (-1)^{k+1} \frac{\acute Y^S_{v[\ul S](S)}}{|A_1|\cdots |A_k|}.
\end{equation}

Lemma \ref{lem:Z}  establishes that $\hat Z^S_v$ and $\acute Z^S$ have the desired margins and belong to $\mscr S^\perp_S$. 

\begin{restatable}{lemma}{lemZ}\label{lem:Z}
    Let $S\in \mscr D$ and assume that the $\hat Y^{\ul S}$'s are self-consistent for all margins $\ul S\subsetneq S$. Let $\hat Z^S$ be constructed as in \eqref{eq:Zhat} and $\acute Z^S$ be from \eqref{eq:Zacute}. Then we have that 
    \begin{enumerate}
        \item for all $\ul S\subsetneq S$ and all $v\in V^{\ul S}$, $\hat Z^S_{v(S)}=\hat Y^{\ul S}_v$ and $\acute Z^S_{v(S)}=\acute Y^S_{v(S)}$, and 
    \item both $\hat Z^S,\acute Z^S\in \mscr S_S^\perp$.
    \end{enumerate}
\end{restatable}

    \begin{algorithm}[t]
        \begin{algorithmic}[1]
            \REQUIRE (A1) holds; $\acute Y^S$ and $(\acute \sigma^2)^S$ for each $S\in \mscr D$ as computed from the collection step
            
            \STATE Set $\hat{Y}^\emptyset = \acute Y^\emptyset$, which is the total estimate


            \FOR{$i=1,2,\ldots, \max_{S\in \mscr D} |S|$ (iteratively go through $i$-way margins)}
            \FORALL{$S\in \mscr D$ such that $|S|=i$ (pick a set of $i$-way margins)}
            \STATE Let $\hat Z^S$ be the table constructed in \eqref{eq:Zhat}
            \IF{(A1)-(A3) hold}
            \STATE Let $\acute Z^S$ be the table constructed in \eqref{eq:Zacute}
            \STATE Set $\hat Y^S = \acute Y^S-\acute Z^S+\hat Z^S$  (by Theorem \ref{thm:FastProj})
            \ELSIF{(A1)-(A2) hold}
            \STATE Call $\acute \Sigma=\mathrm{diag}((\acute \sigma^2)^S)$, which is the covariance matrix of the elements of $\acute Y^S$
            \STATE Set $
                \hat Y^S = \mathrm{Proj}_{ \mscr S_S}^{\acute \Sigma^{-1}}(\acute Y^S) + \hat Z^S
                -\mathrm{Proj}_{\mscr S_S}^{\acute \Sigma^{-1}}(\hat Z^S)$  (projections are no longer orthogonal)
            \ENDIF
                
                \ENDFOR
                \ENDFOR
        \end{algorithmic}
         \caption{Pseudo Code for Down Pass}
                \label{alg:downpass}
    \end{algorithm}

Not only do the constructions $\hat Z^S$ and $\acute Z^S$ have the desired margins and live in $\mscr S^\perp_S$, but in fact, $\acute Z^S$ is the projection of $\acute Y^S$ onto $\mscr S^\perp_S$. This allows us to write the elegant formula $\hat Y^S = \acute Y^S-\acute Z^S+\hat Z^S$ established in the following Theorem:

\begin{restatable}
    {thm}{thmFastProj}\label{thm:FastProj}
    Assume (A1) holds, let $\acute Y^S$ be a table achieved after the collection step, and let $\hat Z^S$ be the result of \eqref{eq:Zhat} using $\hat Y^{\ul S}$ for ${\ul S}\subsetneq S$. Let $\acute Z^S$ be defined as in \eqref{eq:Zacute} using the values in $\acute Y^S$. Then,
    \begin{enumerate}
        \item $\acute Z^S = \mathrm{Proj}^I_{\mscr S^\perp_S} \acute Y^S$ (equivalently, $\acute Y^S-\acute Z^S = \mathrm{Proj}^I_{\mscr S_S} \acute Y^S$), and
        \item if (A2) and (A3) hold, then $\hat Y^S = \acute Y^S-\acute Z^S+\hat Z^S$.
    \end{enumerate}
\end{restatable}

 The intuition behind the formula $\hat Y^S = \acute Y^S-\acute Z^S+\hat Z^S$ is that $\acute Y^S$ has the correct dependence structure, but the wrong margins; so, we ``subtract off the margins'' by subtracting $\acute Z^S$ and then ``add on the correct margins'' by adding $\hat Z^S$.

In Example \ref{ex:Z}, we illustrate how the terms in $\hat Z^S$ cancel when calculating margins. 
\begin{example}\label{ex:Z}
    Let $S = \{A,B,C\}$, and suppose that the $\hat Y^R$'s are self-consistent for all $R\subsetneq S$ (e.g., $\hat Y^{A,B}_{a,\bullet}=\hat Y^A_a$). The construction of $\hat Z^S$ in \eqref{eq:Zhat} is of the form
    \begin{align*}
        \hat Z^S_{a,b,c} &= \left(\frac{\hat Y^{A,B}_{a,b}}{|C|} + 
        \frac{\hat Y^{A,C}_{a,c}}{|B|} + \frac{\hat Y^{B,C}_{b,c}}{|A|}\right)-\left(\frac{\hat Y^{A}_{a}}{|B||C|} + \frac{\hat Y^{B}_b}{|A||C|} + \frac{\hat Y^{C}_{c}}{|A||B|}\right)+\frac{\hat Y^\emptyset}{|A||B||C|}.
    \end{align*}
    If we marginalize over the variable $C$, we obtain:
    \begin{align*}
        \hat Z^S_{a,b,\bullet} &=\left(\hat Y^{A,B}_{a,b} + \frac{\hat Y^{A}_a}{|B|}+\frac{\hat Y^{B}_b}{|A|}\right)-\left(\frac{\hat Y^{A}_a}{|B|} + \frac{\hat Y^{B}_b}{|A|} + \frac{\hat Y^\emptyset}{|A||B|}\right)+\frac{\hat Y^{\emptyset}}{|A||B|},
    \end{align*}
    and we see that all terms cancel except for $\hat Y^{AB}_{a,b}$ as claimed in Lemma \ref{lem:Z}.
\end{example}

\begin{remark}\label{rem:downpass}
    For each step of the down pass, the projection only requires calculating the current margins of $\acute Y^S$ and the cost to assemble $\acute Z^S$ and $\hat Z^S$. In fact, the computational and memory cost for the down pass is of the same rate as the collection step. 
\end{remark}

Algorithm \ref{alg:downpass} gives a pseudo-code implementation of the down pass. Note that while each set of margins must be computed in order, margins with the same number of variables can be calculated in parallel.

\begin{example}\label{ex:toy3}
We continue Example \ref{ex:toy2} by now applying the down pass. First, we calculate $\acute Z^S=(32/3,32/3,32/3)^\top$, since $32=\acute Y^B_\bullet$ and $|B|=3$, and $\hat Z^S=(29.75/3,29.75/3,29.75/3)^\top$, since $\hat Y^\emptyset=\acute Y^\emptyset=29.75$ and $|B|=3$. Thus, 
\[\hat Y^B = \acute Y^B-\acute Z^S+\hat Z^S=
\left(\begin{array}{c}6\\9\\17\end{array}\right)-
\left(\begin{array}{c}32/3\\32/3\\32/3\end{array}\right)+
\left(\begin{array}{c}29.75/3\\29.75/3\\29.75/3\end{array}\right) 
= \left(\begin{array}{c}5.25\\8.25\\16.25\end{array}\right),\]
which agrees with the BLUE calculated in Example \ref{ex:toy}.
\end{example}

\subsection{Complexity of SEA BLUE and Matrix Projection}\label{s:complexity}
In this section, we give a detailed analysis of the computation and memory complexity of the matrix projection in Theorem \ref{thm:general} compared to SEA BLUE. For simplicity, we restrict our attention to the case where noisy tables are observed for all possible margins based on $k\geq 1$ variables with $I\geq 2$ levels each and assumptions (A1)-(A3) hold. In this case, we can calculate that there are $n=(I+1)^k$ noisy counts observed. We consider two asymptotic regimes where either $k$ or $I$ goes to infinity while the other is held fixed. 

{\bf Matrix Projection: } The projection in Theorem \ref{thm:general} requires $O(np+p^{2.373})$ operations to implement \citep{alman2021refined} and $O(np)$ units of memory, where $p$ is the smaller dimension of either the columnspace or nullspace of $A$ ($A$ is defined in Theorem \ref{thm:general}). 

\begin{restatable}
    {lemma}{lemP}\label{lem:P}
    Given all possible tables from $k$ variables, each with $I$ levels, 
    \begin{enumerate}
    \item $p = \min\{I^k,(I+1)^k-I^k\},$
    where $p$ is defined as above, and
    \item matrix projection has runtime  $O(I^{2k-1}+I^{2.373(k-1)})$ and memory $O(I^{2k-1})$ as $I\rightarrow \infty$, and runtime $O(I^k(I+1)^k+I^{2.373k})$ and memory $O(I^k(I+1)^k)$ as $k\rightarrow \infty$.
    \end{enumerate}
\end{restatable}

{\bf SEA BLUE: } In both the collection step and the down pass, it can be seen that the memory requirement never surpases a constant multiple times $n$, a large improvement over the rate $O(pn)$ for the matrix projection. 

In terms of the runtime, for the collection step, the bottleneck is the time to compute all margins from all of the tables. In Lemma \ref{lem:Ztime}, found in the Supplementary Materials, it is established that under (A1)-(A3), the runtime of the down pass is limited by the same quantity. In Lemma \ref{lem:CollectionCost}, we derive the computational cost and memory to compute all margins given either a single table, or all tables from $k$ variables.

\begin{restatable}
    {lem}{lemCollectionCost}\label{lem:CollectionCost}
    \begin{enumerate}
    \item Let $S$ be a table with $k$ variables, each with $I$ levels. Then it is possible to compute all of its margins in $O(I^k)$ time as $I\rightarrow \infty$ and $O((I+1)^k)$ time as $k\rightarrow \infty$; the memory requirement is of the same order.
    \item Given all possible tables from $k$ variables, each with $I$ levels, it is possible to compute all of its margins in $O(I^k)=O(n)$ time as $I\rightarrow \infty$ and $O((2I+1)^k)=O(n^{\log(2I+1)/\log(I+1)})$ time as $k\rightarrow \infty$; under (A1)-(A3), the memory requirement of SEA BLUE is $O(n)$.
    \end{enumerate}
\end{restatable}

\begin{remark}
    In Lemma \ref{lem:CollectionCost}, the quantity $\log(2I+1)/\log(I+1)$ is a decreasing function for $I\geq 0$, which approaches 1 as $I\rightarrow\infty$. At $I=2$ (the smallest number of levels for a variable), it evaluates to $\log(5)/\log(3) \approx 1.46497$. The runtime as $k\rightarrow \infty$ is at worst $o(n^{1.5})$ and approaches $O(n)$ for larger $I$. On the other hand, as $I\rightarrow \infty$, we have runtime of $O(n)$ no matter the value of $k$. 
\end{remark}

 { \bf Comparison Between Matrix Projection and SEA BLUE: }
    We see a large improvement in the memory usage of SEA BLUE compared to the matrix projection: $O(n)$ versus $O(np)$, where $p$ is only slightly smaller than $n$. 
    
    In terms of runtime, when $I\rightarrow \infty$, we see a clear improvement to $O(I^k)$ versus $O(I^{2k-1}+I^{2.373(k-1)})$; these runtimes only agree when $k=1$. As $k\rightarrow \infty$, we can compare  
    $(2I+1)^k\leq 2^k(I+1)^k=O(I^k(I+1)^k)$, 
    where the big-$O$ is tight when $I=2$ and is little-$o$ for $I>2$, 
    and $(2I+1)^k=o(I^{2.373k}),$ for $I\geq 2$. Less precisely, we have that SEA BLUE has runtime between $\Omega(n)$ and $o(n^{1.5})$, whereas the matrix projection has a runtime approximately on the order of $O(n^2)$-$O(n^{2.373})$.

 \section{Covariance Calculation and Confidence Intervals}\label{s:errorQ}
 When publishing differentially private data, it is important to communicate the error due to privacy in the published estimates. In this section, we show that for any linear procedure (including SEA BLUE), we can either calculate the standard error exactly or approximate it with a Monte Carlo approach. In either case, we show how to construct confidence intervals. We also present a distribution-free method for Monte Carlo confidence intervals.

We introduce additional notation to simplify the presentation of this section:
for a table $Z^S$, we let $\vect(Z^S)$ denote a vectorization of the entries of $Z^S$, where we assume a fixed ordering of the variables and levels of the variables. 
 \subsection{Exact Covariance Calculation and \texorpdfstring{$z$}{z}-Confidence Intervals}\label{s:confidence}
 Under assumption (A1)-(A3), we know that SEA BLUE produces the same result as the best linear unbiased estimate described in Theorem \ref{thm:general}. In Theorem \ref{thm:general}, the covariance of the final estimates $\hat Y$ can be obtained by applying the projection procedure to a factorization of the original covariance: $\Sigma^{1/2}$. If the noisy counts $\twid Y$ are normally distributed, we can easily construct confidence intervals for the true values $Y$ using the calculated variances. 

 \begin{restatable}{prop}{propcovarianceExact}\label{prop:covarianceExact}
     Suppose that (A1)-(A3) hold. Let $\Sigma_i^{1/2}$ be the $i^{th}$ column of $\Sigma^{1/2}$. Apply SEA BLUE to $\Sigma_i^{1/2}$ (pretending the vector $\Sigma_i^{1/2}$ contains noisy counts), setting all invariants to zero,   and call the result $\hat \Sigma_i^{1/2}$. Let $\hat\Sigma^{1/2}$ be the matrix whose $i^{th}$ column is $\hat\Sigma^{1/2}_i$. Then, $\hat \Sigma=\hat \Sigma^{1/2}(\hat\Sigma^{1/2})^\top$ is the covariance matrix of $\vect(\hat Y)$. 
     If $\Sigma=\sigma^2I$ (a stronger version of (A3) where all counts have equal variance), then $\hat\Sigma$ is obtained column-wise by simply passing each column of $\Sigma$ through SEA BLUE and no matrix multiplication is needed. 
     
     If $\twid Y^S_v$ are all normally distributed, then a $(1-\alpha)$-confidence interval for $\vect(Y)_i$ is $\vect(\hat Y)_i \pm z_{1-\alpha/2} \sqrt{(\hat\Sigma)_{i,i}}$. 
 \end{restatable}


 \subsection{Monte Carlo Confidence Intervals}
 In Proposition \ref{prop:covarianceExact}, we saw how to compute the exact covariance of $\hat Y$ under assumptions (A1)-(A3). However, if either the assumption of uncorrelated noise (A2) or the assumption of equal variance within tables (A3) are violated, then this result no longer holds. Furthermore, if the variance of every count is desired, running SEA BLUE for $n$ times may be overly expensive.  Finally, even if Proposition \ref{prop:covarianceExact} is applicable, if the $\twid Y$ are not normally distributed,  producing tight CIs based on the covariance may be challenging, as the use of suboptimal concentration inequalities such as Chebyshev and Gauss can result in confidence intervals with excess width.

 To address these limitations, we propose two Monte Carlo confidence interval methods, one under the assumption of normality and the other which holds without any distributional assumptions. These methods simply rely on the fact that the SEA BLUE algorithm is a linear procedure (even if (A2) and (A3) do not hold), and the noise distribution of $N$ is known (as mentioned under ``Immunity to Post-Processing'' in Section \ref{s:DP}).

 \begin{restatable}{prop}{propMonteCarloT}\label{prop:MonteCarloT}
     Suppose that (A1) holds and $\twid Y$ are normally distributed. Then, it follows that $\hat Y$ are also normally distributed with mean $Y$. Let $\vect(\twid Y)^{(1)},\ldots\vect(\twid Y)^{(R)}$ be i.i.d. samples with the same distribution as $\vect(N)$. Call $(\hat \sigma^2)^R_i = \frac 1m\sum_{j=1}^R [\vect(\hat Y^{(j)})_i]^2$, where $\hat Y^{(j)}$ is the result of applying SEA BLUE to $\twid Y^{(j)}$. Then, $\vect(\hat Y)_i \pm t_{1-\alpha/2}(R) \sqrt{(\hat\sigma^2)_i^R}$ 
     is a $(1-\alpha)$-confidence interval for $\vect(Y)_i$. 
 \end{restatable}

 If the $\twid Y$'s are not normally distributed, then we propose a distribution-free Monte Carlo confidence interval using order statistics.

 \begin{restatable}{prop}{propMonteCarloDF}\label{prop:MonteCarloDF}
     Suppose that (A1) holds. Let $\vect(\twid Y)^{(1)},\ldots\vect(\twid Y)^{(R)}$ be i.i.d. samples with the same distribution as $\vect(N)$. Call $\hat Y^{(j)}$ the result of applying SEA BLUE to $\twid Y^{(j)}$. Given an index $i$, let $X^i_{(1)},\ldots, X^i_{(R)}$ be the order statistics of  
     $|\vect(\hat Y^{(1)}_i)|,\ldots, |\vect(\hat Y^{(R)}_i)|. $ 
     If $R\geq (1-\alpha)/\alpha$, then
     $\vect(\hat Y)_i \pm X^i_{(\lceil (1-\alpha) (R+1)\rceil)}$ 
     is a $(1-\alpha)$-confidence interval for $\vect(Y)_i$, where $\lceil\cdot\rceil$ is the ceiling function. 
 \end{restatable}

In Propositions \ref{prop:MonteCarloT} and \ref{prop:MonteCarloDF}, the i.i.d. samples do not use any additional privacy loss budget, as they are simply realizations from the noise distribution and do not access the sensitive data. 
 
 Note that in Proposition \ref{prop:MonteCarloDF}, the Monte Carlo samples are used most efficiently when $(1-\alpha)(R+1)$ is an integer. For example, if $\alpha=.05$, then $R=19$ is the smallest value that can be used in Proposition \ref{prop:MonteCarloDF}, and $R=99$ or $R=199$ are other sensible choices to reduce the expected width of the confidence interval.

In Table \ref{tab:relativeWidth}, we compare the average relative width of the confidence intervals from Propositions \ref{prop:covarianceExact}, \ref{prop:MonteCarloT}, and \ref{prop:MonteCarloDF} in the case of normally distributed data. Compared to using the exact covariance, Proposition \ref{prop:MonteCarloDF} results in about a 10\% increase in expected width when using the smallest value of $R=19$, and this error can be reduced to $<1\%$ by using $R=199$. On the other hand, with Proposition \ref{prop:MonteCarloT}, we can obtain $<1\%$ error with only $R=99$. Finally, when comparing the two Monte Carlo methods to each other, the distribution-free approach has $4\%$ increased width compared to Proposition \ref{prop:MonteCarloT} with $R=19$, and less than $1\%$ increased width at $R=99$. We see that these Monte Carlo methods offer a competitive alternative to Proposition \ref{prop:covarianceExact}.

\begin{table}[!ht]
    \centering
    \begin{tabular}{r|lrrr}
     \multicolumn{2}{r}{$R:$}&19&99&199\\\hline
      Monte Carlo T / z  & &1.055& 1.010 &1.005\\
       Monte Carlo DF / z&  &1.094 &1.017&1.009\\
       Monte Carlo DF / Monte Carlo T &&1.038 & 1.007&1.005
    \end{tabular}
    \caption{Mean relative width of the confidence intervals from Proposition \ref{prop:covarianceExact} (z), Proposition \ref{prop:MonteCarloT} (Monte Carlo T), and Proposition \ref{prop:MonteCarloDF} (Monte Carlo DF). Simulation was conducted with normally distributed random variables, $\alpha=.05$, and 100,000 replicates. In all settings, the Monte Carlo standard errors are $<.001$.}
    \label{tab:relativeWidth}
\end{table}

\begin{remark}
    Suppose we want a confidence interval for multiple counts in $Y$. With either Proposition \ref{prop:MonteCarloT} or Proposition \ref{prop:MonteCarloDF}, the same $R$ copies of $\hat Y^{(j)}$ can be re-used to make a confidence interval for each count of interest. On the other hand, using Proposition \ref{prop:covarianceExact} requires running SEA BLUE for each count of interest that has a different covariance column. So, while the Monte Carlo methods have an overhead of $R$ runs, the number of runs need not scale with the number of desired confidence intervals. 
\end{remark}

\begin{remark}[Refining Confidence Intervals]\label{rem:clip}
While SEA BLUE does not use assumptions of non-negativity and integrality,  we can use these constraints to refine any of the confidence intervals presented earlier in this section without affecting the coverage. For example, if it is known that the true values are non-negative integers, then if the initial confidence interval is $[-1.5,12.7]$ a refined confidence interval would be $[0,12]$. In general, if $[a(\twid Y),b(\twid Y)]$ is a $(1-\alpha)$-confidence interval for a count $Y^S_v$, then so is $[a^*(\twid Y),b^*(\twid Y)]$, where $a^*(\twid Y) = \max\{0,\lceil a(\twid Y)\rceil\}$ and $b^*(\twid Y) = \lfloor b(\twid Y)\rfloor$. 
This method of refining always preserves the coverage and can only reduce the width of the intervals. 
\end{remark}


 \section{Simulations and Application to Census Products}\label{s:simulations}
 In this section, we empirically explore the properties of SEA BLUE in both simulations as well as through an application to two Census products. The code for these experiments is available at \url{https://github.com/JordanAwan/SeaBlue}.

 \subsection{Description of Census Products and Simulation Setup}\label{s:products}
The two U.S. Census Bureau products that we consider are the noisy measurement files for Redistricting Data (PL94) and Demographic and Housing Characteristics (DHC). 

PL94  provides tables summarizing population totals as well as counts for race, Hispanic origin, and voting age for each geographic areas as well as housing unit counts by occupancy status and population counts in group quarters. Structurally, the noised PL94 data consists of four variables, sized 2, 2, 8, and 63. The data for each PL94 geography includes contingency tables for 1) the values for each individual variable by itself, 2) values for the full 4-variable cross combination of all variables, 3) three 2-way crosses (both size-2 variables crossed together, and the size-8 variable crossed against each of the size-2 variables), and 4) a single 3-way cross (size-2, 2, and 8 variables crossed together).

The DHC product provides tables on the following variables: age, sex, race, Hispanic or Latino origin, household type, family type, relationship to householder, group quarters population, housing occupancy, and housing tenure. The contingency tables of the DHC that we used as input included 1) a full-cross of five variables, of sizes 2, 2, 42, 63, and 116, 2) a 2-way cross of the size-2 variables, and 3) a 2-way marginal variable by itself. The DHC data set included several contingency tables that are a cross of variables against ``bins'' of another variable, which weren't directly usable under our current approach. Instead of crossing sex against each possible age value, the data crosses sex against a modified age variable that expresses the counts within large age ranges, e.g., 0-18, 19-36, etc. Our SEA BLUE approach for DHC summed the counts for these sex-cross-binned-age, eliminate these unusable ``binned'' counts, in order to get the sex variable by itself.

 We use the 2010 demonstration products for PL94 and DHC, rather than the official 2020 releases. These demonstration products were protected using DP in a similar manner as the official 2020 releases and were produced to help data users understand the DP mechanism proposed for the 2020 products. We use these demonstration products in our simulations because  we are able to use the Census 2010 Summary File 1 (SF1) counts as the approximately true values for the same geographies. SF1 is the official 2010 tabulation of Census counts and these SF1 values are used to assess mean squared error (MSE) and interval coverage, whereas no analogous product was available for the 2020 release. We also used the 2010 DHC demonstration product to measure timing and memory use for a large contingency table, but we unfortunately did not have the corresponding values to measure the accuracy of DHC results as the 2010 SF1 tables did not match the set of noisy tables in the DHC demonstration product.

 We used a sample    of 40 Census blocks from Rhode Island for PL94, and for DHC, we used state-level counts from Rhode Island. For PL94, our replications were simulated by taking a sample of geographies, since each geography only has one set of counts. Note that within each product, all geographies have the same set of tables, so timing and memory should be similar for different geographies (e.g., SEA BLUE applied to PL94 for a block in Rhode Island should see similar performance to SEA BLUE applied to a county in California). We also ran SEA BLUE on West Virginia state and its 55 counties for PL94, treating county as another variable. We initially chose Rhode Island because of its small size, however, in the end the runtime for SEA BLUE and matrix projection do not depend on the size. We chose West Virginia for multiple-county execution because it was the state with the largest number of counties with uniform variance in the noisy PL94 product. 

 For PL94, all tables at a single geography satisfy assumptions (A1)-(A3). However, DHC does not directly fit within our framework as the Age variable is binned at different thresholds in different tables. Due to this, we marginalized out the Age variable whenever binning occurred and used these tables as the input to SEA BLUE and matrix projection. After this pre-processing, assumptions (A1)-(A3) hold for DHC as well.  
 For the West Virginia counties application, each county had the same privacy loss budget, so (A1)-(A3) held for this example as well, including County as an additional variable; however in other states, counties received different privacy loss budgets, which  results in unequal variances.

 Simulations were conducted in Python, where a custom script was written to implement SEA BLUE, and \texttt{numpy} functions were used to perform the matrix operations needed to implement the matrix projection. Tests were run on a Windows laptop with an Intel i7-7820HQ CPU (4 cores / 2.90GHz) and 8 GB RAM. No parallelization was implemented. 

 In our controlled settings, we generate data for a ``$k\times k$'' table for $k \in {3, 4, 5, 6}$, meaning that there are $k$ variables which each have $k$ levels, and we observe all possible noisy margins. For a single geography we generate the detailed table using a zero inflated Poisson distribution, and aggregate back to the total count.  Noise is added using a normal distribution to each count type independently. 
 \begin{table}[t]
\centering
    \begin{tabular}{rrrrrr}
        \hline
        \multicolumn{2}{c}{Size} & \multicolumn{2}{|c|}{Time} & \multicolumn{2}{c}{Memory}\\
        \hline
        Dimension & Counts & SEA BLUE & Projection & SEA BLUE & Projection \\ 
        \hline
        $3\times3$ &  64 & 0.006 & 0.01 & 2.52 & 6.41 \\ 
        $4\times4$ & 625 & 0.05 & 0.024 & 2.53 & 12.92 \\
        $5\times5$ & 7,776 & 0.90 & 60.45 & 4.15 & 947.54 \\ 
        $6\times6$ & 117,649 & 18.70 & N/A & 30.98 & N/A \\
        PL94 Block & 5,184 & 0.49 & 70.55 & 4.30 & 421.77 \\ 
        PL94 Multi & 290,304 & 34.36 & N/A & 116.16 &  N/A \\
        DHC State & 2,897,856 & 299.07 & N/A & 1186.35 & N/A \\
        \hline
    \end{tabular}
    \caption{Comparison of time (in seconds) and memory (in MiB) performance for SEA BLUE versus full projection on tables of various sizes. Time results are averaged over 40 replicates, and memory is averaged over 10 replicates. PL94 block numbers are based on counts from 40 Rhode Island Census blocks, PL94 Multi numbers are based on counts for West Virginia state and its 55 counties. DHC state is based on Rhode Island state counts. Projection did not run for $6\times 6$, DHC, or multi-geography PL94 due to memory limitations. Only a single instance of SEA BLUE was run on DHC as a test of feasibility.}
    \label{tab:timeMem}
\end{table}
 \subsection{Runtime and Memory Usage}
We implement both SEA BLUE and the matrix projection and apply them to problems of varying sizes. At each size, we measure the runtime as well as the memory usage as either can potentially be a limiting factor. The results are found in Table \ref{tab:timeMem}.

 The second main column of Table \ref{tab:timeMem} contains the runtime in seconds, averaged over 40 replicates. For the $3\times 3$ and $4\times 4$ cases, SEA BLUE and projection  both run in less than a second. However, for the $5\times 5$ problem, SEA BLUE still runs in under 1s, whereas projection takes 60s. Finally, SEA BLUE takes about 19s to process the $6\times 6$ table, whereas projection was not able to be run for this table size due to memory limitations.

The final column of Table \ref{tab:timeMem} records the maximum memory required for the method, measured in mebibytes (MiB), averaged over 10 replicates  (each memory result never varied beyond 3 KiB). We see a substantially improved scaling for memory with SEA BLUE, requiring approximately 4MiB for the $5\times 5$ size compared to nearly 1GiB for the projection. Finally, the $6\times 6$ only required 31MiB for SEA BLUE, whereas the projection was unable to run on our machine due to memory limitations. An error message said that it needed to allocate approximately 103 gibibytes (GiB) for the projection matrix. 

We also applied SEA BLUE and the matrix projection to the Decennial Census 2010 demonstration product of PL94, which is protected with  zero-concentrated differential privacy.  Among the simulations, the PL94 contingency table count is most similar in size to the $5\times 5$ case. Performance-wise, we see SEA BLUE has similar performance on PL94 as in the $5\times 5$ case:  the runtime of SEA BLUE for PL94 is about a $95\%$ decrease compared to the matrix projection and the memory is a $99\%$ decrease compared to the projection. When processing many geographies, this difference in performance becomes practically significant. When running SEA BLUE for PL94 for West Virginia state and its 55 counties simultaneously, the projection failed to run while attempting to allocate 624 GiB of memory.

For the larger DHC product, SEA BLUE's computation took less than 1.2 GiB of memory and less than 30min, while projection failed after attempting to allocate 10TiB of memory. We only ran a single instance of SEA BLUE on DHC as a test of feasibility.

\subsection{Unequal Variances}\label{s:unequal}

\begin{table}[t]
\centering
\begin{tabular}{lrrrr}
  \hline
  Experiment & Counts & SEA & Proj & SEA To Proj \\ 
  \hline 
  One Marginal & 4 & 0.8132 & 0.8131 & 0.0000 \\ 
  All Marginals & 16 & 0.8191 & 0.8191 & 0.0002 \\ 
  All 2 Way & 96 & 0.7995 & 0.7880 & 0.0125 \\ 
  All 3 Way & 256 & 0.7983 & 0.7606 & 0.0408 \\ 
  Detailed & 256 & 0.8159 & 0.8069 & 0.0109 \\ 
  All Counts 1 Var & 500 & 0.7932 & 0.7456 & 0.0451 \\
   \hline
\end{tabular}
\label{table:assume}
\caption{Performance versus true projection in various $4\times 4$ scenarios where (A3) is violated, aggregated over 100 replicates.  Three MSE levels are shown: the SEA BLUE estimates to the true counts, projection values to the true counts, and the SEA BLUE estimates to the projection.}
\end{table}
Recall that SEA BLUE is the best linear unbiased estimate only if the assumptions (A1)-(A3) are met. In the PL94 demonstration product, (A1)-(A3) hold for any single geography. However, it may be the case that different counties in the same state have different variances. In this section, we explore through simulations how much error is incurred by using SEA BLUE when the assumption of equal variance within tables (A3) is violated. 

 We conducted a variety of experiments using a $4\times 4$ simulated contingency tables where a varying number of tables violated (A3). The default variance for each count is set to 2, while in the listed tables the variance was set to be either 1, 2, or 3, uniformly at random.  Experiments were conducted where one marginal table, all marginal tables, all 2-way tables, all 3-way tables, the detailed query, or all counts containing a single variable each received non-equal variance. In each case, the MSE was calculated using 100 replicates (because of the small size of the $4\times 4$ tables, we could afford additional replicates here).  
 
The higher proportion of counts that fall within non-uniform tables, the greater the error that SEA BLUE incurs compared to the ideal projection. Nevertheless, these mild violations of (A3) still result in near-optimal MSE (off by about 5\% in the worst case).  For reference, the MSE of the original counts is approximately 2 in all settings.  


\subsection{Confidence Interval Validation}

\begin{table}[t]
    \centering
    \begin{tabular}{r|r|r|r}
        Scenario & Exact & Monte Carlo T& Monte Carlo DF\\
        \hline
        $3\times3$ & 0.9531 & 0.9480 & 0.9582 \\
        $4\times4$ & 0.9520 & 0.9516 & 0.9541 \\
        $5\times5$ & 0.9510 & 0.9500 & 0.9520 \\
        $6\times6$ & 0.9503 & 0.9498 & 0.9523 \\
        PL94 Block & 0.9505 & 0.9510 & 0.9543 \\
        PL94 Multi & 0.9482 & 0.9490 & 0.9733
    \end{tabular}
    \caption{Coverage shown for various scenarios based on 40 replicates using 20 rounds of estimation per trial. $N\times N$ replicates are computed on replicate-generated random noisy data. PL94 runs are based on 40 block-level records from the state of Rhode Island. $\alpha = 0.95$. The Exact method is from Proposition \ref{prop:covarianceExact}, Monte Carlo T is from Proposition \ref{prop:MonteCarloT}, and Monte Carlo DF is from Proposition \ref{prop:MonteCarloDF}.}
    \label{tab:coverage}
\end{table}

In this section, we numerically validate the confidence interval methods of Section \ref{s:errorQ}. In particular, we estimate the coverage and width on both simulated data and the PL94 Census data product. Confidence interval testing against DHC was not possible because reference values were not available for the DHC tables.

In Table \ref{tab:coverage}, we measure the empirical coverage of our three confidence interval methods on both simulated data and on the PL94 product averaged over 40 block-level records from Rhode Island. Note that in the simulated tables, the noise distribution was continuous Gaussian whereas for PL94 the noise distribution was discrete Gaussian. Even with the slight deviation from normality, we still approach the 95\% nominal coverage for all of our proposed methods. The over-coverage of the distribution-free approach is to be expected, since this approach makes fewer assumptions on the noise distribution.
\begin{table}[t]
\centering
    \begin{tabular}{r|rr|rr|rr|rr}
        \hline
        \multicolumn{1}{c|}{Scenario} & \multicolumn{2}{c|}{Initial} & \multicolumn{2}{c}{Exact} & \multicolumn{2}{|c|}{T Monte Carlo} & \multicolumn{2}{c}{DF Monte Carlo} \\
        \hline
         & Raw & Clip & Raw & Clip & Raw & Clip & Raw & Clip \\ 
        \hline 
        $3\times3$ & 5.544 & 3.413 & 3.601 & 2.240 & 3.772 & 2.361 & 3.963 & 2.546 \\ 
        $4\times4$ & 5.544 & 3.506 & 3.548 & 2.254 & 3.728 & 2.407 & 3.915 & 2.574  \\
        $5\times5$ & 5.544 & 3.549 & 3.514 & 2.242 & 3.694 & 2.400 & 3.884 & 2.568 \\
        $6\times6$ & 5.544 & 3.568 & 3.491 & 2.233 & 3.669 & 2.392 & 3.859 & 2.5605 \\
        PL94 Block & 43.802 & 21.765 & 19.342 & 9.319 & 20.300 & 9.831 & 21.349 & 10.358 \\
        PL94 Multi & 31.230 & 16.062 & 23.740 & 11.997 & 23.772 & 12.013 & 27.257 & 13.832 \\
        DHC State & 13.044 & 6.421 & 12.887 & 6.121 & 13.544 & 6.371 & 14.241 & 6.727  \\
        \hline
    \end{tabular}
    \caption{Comparison of confidence interval widths for tables that were supplied as input. 
    Intervals included for initial noise, SEA BLUE variance calculation, T-based, and distribution-free Monte Carlo estimations. Results averaged from 40 replicates using 20 rounds of Monte Carlo estimation per replicate. $N\times N$ replicates are computed on replicate-generated random noisy data. PL94 runs are based on 40 block-level records from the state of Rhode Island. DHC numbers are computed from Rhode Island state-level counts.}
    \label{tab:widthObs}
\end{table}

In Table \ref{tab:widthObs}, we compare the width of our confidence interval procedures. In each case, we consider both the raw and clipped intervals, where the clipping is according to Remark \ref{rem:clip}. For the simulated examples, the input and output tables are identical; for PL94 and DHC, the output is all tables possible from the marginals, but the input is the subset of tables actually supplied in the Census data. In all cases, our clipped intervals are significantly smaller than the initial clipped intervals. In the simulated tables and for PL94, we also see the unclipped confidence interval width is significantly reduced. Interestingly, for DHC the unclipped width is not significantly reduced, but the clipped width is much smaller after SEA BLUE. Clipping has the greatest impact when there are many true counts near zero; we suspect that it must be the case that in DHC, there are many such counts, which are disproportionately affected by using the clipping. 

Similarly, Table \ref{tab:widthNobs} compares the widths of our confidence intervals but for counts that do not have original noisy estimates. Note that this only applies to PL94 and DHC since the simulated tables have noisy counts for all margins. Compared to Table \ref{tab:widthObs}, these widths are larger, which makes sense since no direct measurement was available for these counts. Otherwise, we see similar behavior in Table \ref{tab:widthNobs}: the clipped intervals are significantly shorter, and the exact variance method is the smallest, followed by the Monte Carlo T-interval, with the distribution-free Monte Carlo interval being the largest.

\begin{table}[t]
\centering
    \begin{tabular}{r|rr|rr|rr}
        \hline
        \multicolumn{1}{c|}{Scenario} & \multicolumn{2}{c}{Exact} & \multicolumn{2}{|c|}{T Monte Carlo} & \multicolumn{2}{c}{DF Monte Carlo} \\
        \hline
         & Raw & Clip & Raw & Clip & Raw & Clip \\ 
        \hline
        PL94 Block & 20.9643 & 10.0505 & 22.0136 & 10.5767 & 23.1337 & 11.1392 \\
        PL94 multi & 25.0119 & 12.8369 & 25.04334 & 12.8535 & 28.7145 & 14.7892 \\
        DHC State & 26.5573 & 13.6419 & 27.8967 & 14.3599 & 29.3217 & 15.1036 \\
        \hline
    \end{tabular}
    \caption{Comparison of confidence interval widths for tables that were \emph{not} supplied as input. 
    Intervals included for initial noise, SEA BLUE variance calculation, T-based, and distribution-free Monte Carlo intervals size averaged from 40 replicates using 20 rounds of Monte Carlo estimation per replicate. PL94 runs are based on 40 block-level records from  Rhode Island. DHC numbers are computed from Rhode Island state-level counts.}
    \label{tab:widthNobs}
\end{table}

 \section{Discussion and Future Work}
 Under assumptions (A1)-(A3), our SEA BLUE method offers a scalable implementation of the best linear unbiased estimate of true counts, subject to self-consistency constraints. As it is a linear procedure, we can easily produce confidence intervals with provable coverage guarantees. Compared to a naive matrix projection, our method scales significantly better in terms of both runtime and memory requirements and we show that it can be applied to large Census products. 

The proposed method was initially designed to process a subset of the Decennial Census products at a single geography (e.g., a single county or a single block), but we were able to also apply our method to incorporate two levels of geography, such as a state with its counties. However, including more than two levels of geography would require an extension of our framework since, for example, different states would have a different set of counties so the ``county'' variable changes depending on the state. We suspect that a modification of our collection step and down pass approach can accommodate adding in multiple geographies, but we leave this as future work. 

 While Section \ref{s:unequal} showed that SEA BLUE is robust to mild violations to (A3), in general it could be arbitrarily far from the projection estimates when there are more egregious violations of the (A3). Performance could be improved by modifying Equations  \eqref{eq:Zhat} and \eqref{eq:Zacute} to handle unequal variances. However, even with this modification, SEA BLUE would not in general produce the projection when (A3) does not hold. As mentioned in the discussion before Theorem \ref{thm:total}, when only (A1) and (A2) hold, the collection step produces an optimal estimator using the estimates that are ``from below'' and ``within-table''; thus, this step always improves over the original noisy counts, but could be arbitrarily worse than the optimal estimator (Example \ref{ex:A3forCollection} offers an example where SEA BLUE produces an estimator with variance 8, where the projection estimator would have variance at most $3$). Future work could investigate whether SEA BLUE can be modified to offer more robustness to violations of (A3).

 While (A2) is more commonly used in similar literature, and is also satisfied by the 2020 census products, it is also worth considering extensions of SEA BLUE that can handle violations of (A2). As noted in Section \ref{s:collection}, it seems that without (A2), it may be necessary to track the covariances of the estimates, which would significantly increase the computation and memory requirements. Future work may investigate whether weaker versions of (A2) or modifications to the SEA BLUE procedure can maintain computational efficiency while allowing for dependence in the noise distribution.

 As mentioned in Remark \ref{rem:invariant}, this paper does not directly address \emph{invariants}, which are counts observed without noise. We believe that the procedure can be modified to incorporate other linear equality invariants such as these as well.



  SEA BLUE only accommodates linear equality constraints in order to remain a linear and unbiased procedure. This has two consequences: 1) SEA BLUE can result in negative and non-integer-valued estimates and 2) SEA BLUE is expected to have higher MSE than the Top Down Algorithm, which incorporates non-negativity and integrality (see \citet{mccartan2023making} for a comparison of the Top Down Algorithm against the collection step). As noted in Remark \ref{rem:clip}, our confidence intervals can be refined to incorporate additional constraints without affecting coverage. Nevertheless, accuracy could be improved by enforcing additional constraints in the estimation procedure. An important direction for future work would be to develop computationally efficient post-processing methods which incorporate all constraints, preserve (at least approximate) unbiasedness, and also give valid confidence regions.

  The results of Section \ref{s:confidence} are focused on coordinate-wise confidence intervals, but one may also be interested in confidence regions that simultaneously cover multiple counts. Extension of Propositions \ref{prop:covarianceExact} and \ref{prop:MonteCarloT} to give confidence ellipsoids should be straightforward by leveraging the normality assumption. An extension of Proposition \ref{prop:MonteCarloDF} is less straightforward, but could potentially be done using the order statistics of a depth statistic (e.g., \citealp{yeh1997balanced,xie2022repro,awan2024simulation}). However, it is not clear whether these approaches are computationally tractable, especially when making a confidence set for all counts, simultaneously. We leave it to future researchers to investigate the details and feasibility of such confidence regions. 

Our simulations  did not implement parallelization techniques in the implementation of SEA BLUE. For large databases, the ability to parallelize could enable even larger tables to be analyzed by SEA BLUE. It would be worth developing R/Python packages that implement the SEA BLUE method with parallel optimizations incorporated.

\acks{This work was supported by Department of Commerce Contract 1331L523D130S0003 Task Order 1333LB23F00000196. Jordan Awan's research was also supported in part by NSF grant no. SES-2150615 to Purdue University. The authors are grateful to Nathaniel Cady who contributed early numerical work to the project, and to Ryan Janicki, Philip Leclerc, Mikaela Meyer, as well as the anonymous reviewers for their feedback which improved the presentation of this manuscript.} 



\appendix

\section{ Summary of Notation}\label{s:notation}
 We include in Table \ref{tab:notation} a summary of the notation used in this paper, which is defined in Section \ref{s:tableNotation}.

\begin{table}[ht!]
    \centering
    \begin{tabular}{c|c|c}
        Notation & Context & Definition/Interpretation \\ \hline
         $Y$& & true counts\\
         $\tilde Y$ & & observed noisy counts\\
         $\acute{Y}$ & & intermediate estimate of $Y$ \\
         $\hat Y$ && final estimate of $Y$ from SEA BLUE\\
         $\mscr A$&& set of variables in the contingency tables\\
         $\mscr O$& $\mscr O\subset 2^{\mscr A}$& set of observed tables\\
         $\mscr D$&$\mscr D\subset 2^{\mscr A}$& set of desired tables\\
         $|A|$&$A\in \mscr A$& number of levels that variable $A$ can take on\\
         $\tilde Y^S$& $S\subset \mscr A$& a table of noisy counts for variables in $S$\\
         $V^S$ & $S\subset \mscr A$& the set of indices for $Y^S$\\
         $\twid Y^S_v$&$v\in V^S,S\subset \mscr A$& the single noisy count in $\twid Y^S$ indexed by $v$\\
         $V^S_\bullet$&$S\subset \mscr A$& same as $V^S$, but each coordinate also can be``$\bullet$''\\
         $\twid Y^S_v$& $v\in V^S_\bullet,S\subset \mscr A$& index with $v$, summing out variables with ``$\bullet$''\\
         $v(\ol S)$&$\ol S\supset S$, $v\in V^S,S\subset \mscr A$& $v(\ol S)\in V^{\ol S}_{\bullet}$ agrees with $v$ on $S$ and has ``$\bullet$'' on $\ol S\setminus S$\\
         $v[\ul S]$&$\ul S\subset S, v\in V^S,S\subset \mscr A$& $v[\ul S]\in V^{\ul S}$ extracts entries of $v$ corresponding to $\ul S$\\
         $v[\ul S](S)$&$\ul S\subset S, v\in V^S,S\subset \mscr A$& $v[\ul S](S)\in V^S_\bullet$ agrees with $v$ on $\ul S$ and has ``$\bullet$'' on $S\setminus \ul S$.
    \end{tabular}
    \caption{ Summary of notation from Section \ref{s:tableNotation}.}
    \label{tab:notation}
\end{table}
\section{Proofs and Technical Details}\label{s:proofs}

We review some properties of linear transformations of random variables: let $A\in \RR^{n\times k}$ and $b\in \RR^n$ and let $Y\in \RR^k$ be a random variable with mean $\mu$ and covariance $\Sigma$. Then $\EE AY=A\mu$ and $\mathrm{Cov}(AY)=A\Sigma A^\top$. In particular, if  $Y\sim N(\mu, \Sigma)$. Then $AY+b\sim N(A\mu+b, A\Sigma A^\top)$. 

\thmgeneral*
\begin{proof}
Since $A \mu=b$, we know that $b\in \mathrm{columnSpace}(A)$. So, there exists $v$ such that $b=A v$. Now, we can say that $\mu-v\in \mathrm{nullSpace}(A)$. We then transform our problem as follows. Let $X'=\Sigma^{-1/2}(X-v)$, $\mu'=\Sigma^{-1/2}(\mu-v)$. Note that $A\Sigma^{1/2}\mu'=0$, so that $\mu'\in\mathrm{nullSpace}(A\Sigma^{1/2})$. Let $M$ be a matrix whose columns form a basis for $\mathrm{nullSpace}(A\Sigma^{1/2})$. Then there exists a vector $\beta$ such that $\mu'=M\beta$. So, we have rephrased our problem as $X'=M\beta+e$, where $e$ is a mean-zero vector, with covariance $I$. By the Gauss-Markov Theorem, the BLUE for $\beta$ is $\hat\beta=(M^\top M)^{-1}M^\top X'$, which implies that the BLUE for $\mu'$ is $\mathrm{Proj}_{\mathrm{nullSpace}(A\Sigma^{1/2})}^I X'$. Converting back to our original problem, we have that the BLUE for $\mu$ is $\hat \mu=\Sigma^{1/2} \mathrm{Proj}_{\mathrm{nullSpace}(A\Sigma^{1/2})}^I(\Sigma^{-1/2}(X-v))+v$, which can be simplified to $\hat\mu = \mathrm{Proj}_{S_0}^{\Sigma^{-1}}(X)+(\mathrm{Proj}_{S_0^\perp}^\Sigma)^\top v$. Similarly, if $X\sim \mscr N(\mu,\Sigma)$, then $\hat\beta$ is the uniformly minimum variance unbiased estimator for $\beta$, which implies that $\hat\mu$ is the uniformly minimum variance unbiased estimator for $\mu$. 
\end{proof}

\lemcollectionVar*
\begin{proof}
To simplify notation, we write $\sum_{\ol S\supset S}$ in place of $\sum_{\ol S\in \mscr O | \ol S\supset S}$. Recall that $\acute Y^S_v$ can be expressed as
    \[\acute Y^S_v = \frac{\sum_{\ol S\supset S} \left(\var(\twid Y^{\ol S}_{v(\ol S)})\right)^{-1} \twid Y^{\ol S}_{v(\ol S)}}{\sum_{\ol S\supset S} \left(\var(\twid Y^{\ol S}_{v(\ol S)})\right)^{-1}}.\]
    Using the fact that all of the $\twid Y$ values are independent, we calculate the variance as
    \begin{align*}
        \var\left(\acute Y^S_v\right) &=\frac{\sum_{\ol S\supset S} \left(\var(\twid Y^{\ol S}_{v(\ol S)})\right)^{-2} \var\left(\twid Y^{\ol S}_{v(\ol S)}\right)}{\left(\sum_{\ol S\supset S} \left(\var(\twid Y^{\ol S}_{v(\ol S)})\right)^{-1}\right)^{2}}\\
        &=\frac{\sum_{\ol S\supset S} \left(\var(\twid Y^{\ol S}_{v(\ol S)})\right)^{-1}}{\left(\sum_{\ol S\supset S} \left(\var(\twid Y^{\ol S}_{v(\ol S)})\right)^{-1}\right)^{2}}\\
        &=\left(\sum\limits_{\ol S\supset S} \left(\var(\twid Y^{\ol S}_{v(\ol S)})\right)^{-1}\right)^{-1}.
    \end{align*}
\end{proof}

 \thmtotal*
 \begin{proof}
     Let $\mscr T$ be the affine space consisting of all possible linear unbiased estimates of $Y^\emptyset$, based on the counts in $\twid Y^S$ for $S\in \mscr O$. Note that $\{\twid Y^S_{\bullet}\mid S\in \mscr O\}$ is a set of linearly independent members of $\mscr T$, where the subscript ``$\bullet$'' is short hand to represent the vector in $V^S_\bullet$ with ``$\bullet$'' in every entry. Then, there exists a basis vector $b$ for $\mscr T$ containing $\twid Y^S_\bullet$ as the first entries for all $S\in \mscr O$. Let $\Lambda$ be the covariance matrix for $b$. 
     
     To see that $\Lambda$ is positive definite, let $\omega$ be a nonzero vector. Then $\var(\omega^\top b)=\omega^\top \Lambda \omega$. Since $\twid Y^S_v$ are independent with nonzero variance for all $v\in V^S$ and $S\in \mscr O$, and because the entries of $b$ are linearly independent, it follows that there exists a nonzero vector $\twid \omega$ such that $\omega^\top b = \twid \omega^\top \vect(\twid Y)$, where $\vect(\twid Y^S_v)$ represents a vectorization of the entries of $\twid Y^S_v$ for all  $v\in V^S$ and $S\in \mscr O$. Then, we have that 
     \[\var(\twid \omega^\top \vect(\twid Y)) = \sum_{S\in \mscr O}\sum_{v\in V^S} \twid \omega^S_v \var(\twid Y^S_v)> 0,\]
     where $\twid \omega^S_v$ represents the entry of $\twid \omega$ that corresponds with $\twid Y^S_v$. So, we have that $\omega^\top \Lambda \omega=\var(\twid \omega^\top  \vect(\twid Y))> 0$, which implies that $\Lambda$ is positive definite. 

     The best linear unbiased estimator for $Y^\emptyset$ uses the inverse covariance weighting of $b$, which can be expressed in terms of matrix operations as, 
     \[\frac{\ul 1^\top \Lambda^{-1}b}{\ul 1^\top (\Lambda^{-1})\ul 1},\] where $\ul 1$ is the vector of all $1$'s of the appropriate dimension, which is the efficient generalized least squares solution \citep{anderson2008multiple}.

     On the other hand, the collection estimate $\acute Y^\emptyset$ can be expressed as 
     \[\frac{\left( \left(\frac{1}{\var(\twid Y^S_{\bullet})}\right)_{S\in \mscr O},0,\ldots 0\right)b}{\sum_{S\in \mscr O}  \frac{1}{\var(\twid Y^S_\bullet)}},\]
     where $\left(1/\var(\twid Y^S_{\bullet})\right)_{S\in \mscr O}$ represents the vector of inverse-variances in the same order as the elements $\twid Y^S_\bullet$ appear in $b$.

     To claim that the collection step results in the BLUE estimate of $Y^\emptyset$, we need to establish that
     \[\frac{\left( \left(\frac{1}{\var(\twid Y^S_{\bullet})}\right)_{S\in \mscr O},0,\ldots 0\right)}{\sum_{S\in \mscr O}  \frac{1}{\var(\twid Y^S_\bullet)}}
     =\frac{\ul 1^\top \Lambda^{-1}}{\ul 1^\top \Lambda^{-1} \ul 1}.\]
      It suffices to show that 
     \begin{equation}\label{eq:invCov}\left( \left(\frac{1}{\var(\twid Y^S_{\bullet})}\right)_{S\in \mscr O},0,\ldots 0\right)\Lambda=\ul 1^\top,\end{equation}
     as this implies that 
     \[\left( \left(\frac{1}{\var(\twid Y^S_{\bullet})}\right)_{S\in \mscr O},0,\ldots 0\right)=\ul 1^\top\Lambda^{-1},\]
     as well as 
     \[\sum_{S\in \mscr O}  \frac{1}{\var(\twid Y^S_\bullet)}=\left( \left(\frac{1}{\var(\twid Y^S_{\bullet})}\right)_{S\in \mscr O},0,\ldots 0\right)1=\ul 1^\top\Lambda^{-1} 1.\]
     The remainder of the proof aims to establish \eqref{eq:invCov}.

     Let $Z$ be an element of $b$. Then, restricting attention to the column of $\Lambda$ corresponding to $Z$, we need to show that 
     \[\sum_{S\in \mscr O} \frac{\cov(Z,\twid Y^S_\bullet)}{\var(\twid Y^S_\bullet)} =1.\]

     We can express $Z = \sum_{S\in \mscr O} \sum_{v\in V^S} a^S_v \twid Y^S_v$ for some weights $a^S_v$. We claim that for $Z$ to be an unbiased estimator for $Y^\emptyset=\sum_{e\in V^{\mscr A}} Y_e$, it must be that $\sum_{S\in \mscr O}\sum_{\substack{v\in V^S\\v =e[S]}} a^S_v=1$, for all $e\in V^{\mscr A}$ 
     To see this, we calculate
     \begin{align*}
         \EE Z &= \sum_{S\in \mscr O} \sum_{v\in V^S} a_v^S \EE \twid Y_v^S\\
         &=\sum_{S\in \mscr O} \sum_{v\in V^S} a_v^S\sum_{e[S]=v}Y_e\\
         &=\sum_{e\in V^{\mscr A}}Y_e \sum_{S\in \mscr O} \sum_{\substack{v\in V^S\\ v=e[S]}} a_v^S,
     \end{align*}
     where in the second line we used the fact that $\EE \twid Y_v^S = Y^S_v = \sum_{e[S]=v} Y_e$. We see that in order for $\EE Z = Y^\emptyset = \sum_{e\in V^{\mscr A}} Y_e$, it must be that $\sum_{S\in \mscr O}\sum_{\substack{v\in V^S\\v =e[S]}} a^S_v=1$ as claimed.


     Using our expansion of $Z$, we can now calculate 
     \begin{align}
         \sum_{S\in \mscr O} \frac{\cov(Z,\twid Y^S_\bullet)}{\var(\twid Y^S_\bullet)} 
         &=\sum_{S\in \mscr O}\sum_{v\in V^S} \frac{a^S_v \var(\twid Y^S_v)}{|V^S|\var(\twid Y^S_v)}\label{eq:equalvar}\\
         &=\sum_{S\in \mscr O}\sum_{v\in V^S} \frac{a^S_v}{|V^S|}\\
         &=\sum_{S\in \mscr O}\sum_{e\in V^{\mscr A}} \frac{|V^S|}{|V^{\mscr A}|} \sum_{\substack{v\in V^S\\ v= e[S]}} \frac{a_v}{|V^S|}\label{eq:overcount}\\
         &=\frac{1}{|V^{\mscr A}|} \sum_{e\in V^{\mscr A}} \sum_{S\in \mscr O}\sum_{\substack{v\in V^S\\ v= e[S]}} a^S_v\\
         &=\frac{1}{|V^{\mscr A}|} \sum_{e\in V^{\mscr A}} 1\\
         &=1,
     \end{align}
     where in \eqref{eq:equalvar} we use the assumption that the variance is constant within tables (A3); in \eqref{eq:overcount}, the factor $|V^{\mscr A}|/|V^S|$ corrects for over-counting, since there are $|V^{\mscr A}|$ $e$'s and $|V^S|$ $v$'s, and each corresponds to an equal number. We see that for every column, $\Lambda_{\cdot,i}$, we have $((1/\var(\twid Y^S)),0,\ldots, 0) \Lambda_{\cdot,i}=1$ and conclude that $\twid Y^\emptyset$ is the BLUE for $Y^\emptyset$. 
 \end{proof}

\thmOptimal*
 \begin{proof}
    For any $S\in 2^{\mscr A}\setminus \mscr O$, artificially set $\twid Y^S_v=0$ for all $v\in V^S$ and $\sigma^2(\twid Y^S_v)=\infty$. We use the convention that for any non-negative number $a$, $a/\infty= 0$. 

    Let $S\in \mscr D$ be chosen, and assume that for all $T_i\subset S$ such that $|T_i|=|S|-1$, we have the BLUE estimates $\hat Y^{T_i}$ (by Proposition \ref{thm:total}, we can compute $\hat Y^\emptyset$; for induction we can assume that we have access to margins ``above'' $S$). As noted by \citet[Proof of Theorem 3]{hay2010boosting}, when holding $\hat Y^{T_i}$ fixed, the values of $\hat Y^S$ are independent of $\twid Y^R$ whenever $R\not\supset S$, and the BLUE $\hat Y^S$ can be expressed as the solution to the following optimization problem:

    \begin{align}
        &\text{minimize } \sum_{R\supset S} \sum_{v\in V^{R}} (\hat Y^{R}_v-\twid Y^{R}_v)^2/\sigma^2(\twid Y^{R}_v).\label{eq:optim}\\
        &\text{subject to } \hat Y^{\ol{R}}_{v(\ol {R})}=\hat Y^{R}_v,\quad\text{for all $\ol{R}\supset R\supset S$, $|\ol{R}|=|R|+1$, and $v\in V^{R}$},\notag\\
        &\text{and } \hat Y^S_{v(S)}=\hat Y^{T_i}_v,\quad \text{for all $T_i$ and $v\in V^{T_i}$}.\notag
    \end{align}
    In the case that $S=\mscr A$ (full detail), this simplifies to 
        \begin{align*}
        &\text{minimize } \sum_{v\in V^{\mscr A}} (\hat Y^{\mscr A}_v-\acute Y^{\mscr A}_v)^2/\sigma^2(\twid Y^{\mscr A}_v)\\
        &\text{subject to } \hat Y^{\mscr A}_{v(\mscr A)}=\hat Y^{T_i}_v,\quad \text{for all $T_i$ and $v\in V^{T_i}$},
    \end{align*}
    where we used the substitution $\acute Y^{\mscr A}=\twid Y^{\mscr A}$, and we see that in this case, the solution to this minimization problem is the projection prescribed in the down pass step.

    Now assume that $S\neq \mscr A$. We write out the Lagrangian for \eqref{eq:optim}:
    \begin{equation}\label{eq:lagrange}
        \sum_{R\supset S}\sum_{v\in V^{R}} (\hat Y^{R}_v-\twid Y^{R}_v)^2/\sigma^2(\twid Y^{R}_v) +\sum_{\substack{\ol{R}\supset R\supset S\\|\ol{R}|=|R|+1}} \sum_{v\in V^{R}} \lambda^{R,\ol{R}}_v(\hat Y^{\ol{R}}_{v(\ol{R})}-\hat Y^{R}_v)+\sum_{T_i}\sum_{v\in V^{T_i}}\lambda^{T_i,S}_v(\hat Y^S_{v(s)}-\hat Y^{T_i}_v).
    \end{equation}
    Note that this is a convex and smooth objective. To understand the solution, we examine the partial derivatives. Taking the derivative with respect to any $\lambda$ just recovers the original equality constraints. For the other derivatives, we separate them into three cases.

    \noindent Case 1: Let $v\in V^S$ and take the derivative of \eqref{eq:lagrange} with respect to $\hat Y^S_v$ and set equal to zero: 
    \begin{equation}\label{eq:1}
        2(\hat Y^S_v-\twid Y^S_v)/\sigma^2(\twid Y^S_v)-\sum_{\substack{\ol S\supset S\\ |\ol S|=|S|+1}} \lambda^{S,\ol S}_v + \sum_{T_i} \lambda^{T_i,S}_{v[T_i]}=0,
    \end{equation}
    where recall that by (A2), $\sigma^2(\twid Y^S_v)$ does not depend on $v$.

    \noindent Case 2: For any $R\supsetneq S$, $R\neq \mscr A$ and for any $v\in V^{R}$, the derivative with respect to $\hat Y^{R}_v$ gives:
    \begin{equation}\label{eq:2}
        2(\hat Y^{R}_v-\twid Y^{R}_v)/\sigma^2(\twid Y^{R}_v)-\sum_{\substack{\ol{R}\supset {R}\\ |\ol{R}|=|R|+1}} \lambda^{R,\ol{R}}_v + \sum_{\substack{\ul R\subset R,\ul R\supset S\\ |\ul R|=|R|-1}} \lambda^{\ul R,R}_{v[\ul R]}=0.
    \end{equation}
    
    \noindent Case 3: In the case of $\mscr A$ and $v\in V^{\mscr A}$, the derivative with respect to $\hat Y^{\mscr A}_v$ gives:
    \begin{equation}\label{eq:3}
        2(\hat Y^{\mscr A}_v - \twid Y^{\mscr A}_v)/\sigma^2(\twid Y^{\mscr A}) 
        +\sum_{\substack{|\ol S|=|\mscr A|-1\\ \ol S\supset S}} \lambda^{\ol S,\mscr A}_{v[\ol S]} = 0.
    \end{equation}
    Together, all of the instances of \eqref{eq:1}, \eqref{eq:2}, and \eqref{eq:3}, along with the self-consistency constraints describe the solution space. In what follows, we combine these equations to show that the solution space can be decomposed into a down pass to obtain $\hat Y^S$ and a projection to compute the other entries. 

    Let $v\in V^S$ be held fixed. Choose $R\supset S$ such that $|R|=|S|+1$. Summing up \eqref{eq:2} over all $w\in V^R$ which satisfy $w[S]=v$, and dividing through by $|V^R|/|V^S|$ (which is the number of equations being added), gives 
    \[\frac{|V^S|}{|V^R|} 2(\hat Y^S_v - \twid Y^R_{v(R)}) \sigma^2(\twid Y^R_v) -  \sum_{\substack{\ol R\supset R\\ |\ol R|=|R|+1}}\frac{|V^S|}{|V^R|}\sum_{\substack{w\in V^R\\ w[S]=v}} \lambda^{\ol R,R}_w + \lambda^{S,R}_v=0.\]
    Using the fact that $\sigma^2(\twid Y^R_{v(R)}) = (|V^R|/|V^S|)\sigma^2(\twid Y^R_v)$ and introducing the notation $\overline \lambda^{\ol R,R}_v = \frac{|V^S|}{|V^R|}\sum_{\substack{w\in V^R\\ w[S]=v}} \lambda^{\ol R,R}_w $, we simplify this as:
    \begin{equation}
        \label{eq:4}
        2(\hat Y^S_v - \twid Y^R_{v(R)})/\sigma^2(\twid Y^R_{v(R)}) - \sum_{\substack{\ol R\supset R\\ |\ol R|=|R|+1}} \overline \lambda^{\ol R,R}_{v} + \lambda^{S,R}_v=0.
    \end{equation}
    Adding \eqref{eq:4} to \eqref{eq:1} for all $R\supset S$ such that $|R|=|S|+1$, we see that the negative $\lambda$'s in \eqref{eq:1} cancel with the positive $\lambda$'s in \eqref{eq:4}. This gives:
    \begin{equation}
        \label{eq:1.1}
        2\left(\hat Y^S_v\sum_{\substack{R\supset S\\ |R|=|S|+1}} \frac{1}{\sigma^2(\twid Y^R_{v(R)})} - \sum_{\substack{R\supset S\\ |R|=|S|+1}} \frac{\twid Y^R_{v(R)}}{\sigma^2(\twid Y^R_{v(R)})}\right) - \sum_{\substack{\ol R\supset R\supset S\\ |\ol R|=|R|+1=|S|+2}} \overline \lambda^{R,\ol R}_{v} + \sum_{T_i} \lambda^{T_i,S}_{v[T_i]} = 0.
    \end{equation}

    Similarly, we next choose $R\supset S$ such that $|R|=|S|+2$, sum up \eqref{eq:2} over all $w\in V^R$ such that $w[S]=v$, and divide by $|V^R|/|V^S|$:
    \begin{equation}\label{eq:4.1}
        2(\hat Y^S_v - \twid Y^R_{v(R)})/\sigma^2(\twid Y^R_{v(R)}) - \sum_{\substack{\ol R\supset R\\ |\ol R|=|R|+1}} \overline \lambda^{R,\ol R}_{v} + \sum_{\substack{\ul R\subset R, R'\supset S\\ |\ul R|=|R|-1}} \overline \lambda^{\ul R,R}_v=0,
    \end{equation}
    where the last term is obtained by simplifying the following:
    \begin{align}
        \sum_{\substack{\ul R\subset R,\ul R\supset S\\ |\ul R|=|R|-1}}\frac{|V^S|}{|V^R|} \sum_{\substack{w\in V^R\\ w[S]=v}} \lambda^{\ul R,R}_{w[\ul R]}
    &=\sum_{\substack{\ul R\subset R,\ul R\supset S\\ |\ul R|=|R|-1}} \frac{|V^S|}{|V^{\ul R}|} \sum_{\substack{v'\in V^{\ul R}\\ v'[S]=v}} \frac{|V^{\ul R}|}{|V^R|} \sum_{\substack{w\in V^R\\ w[\ul R]=v'}} \lambda^{\ul R,R}_{v'}\label{eq:simplify}\\
    &=\sum_{\substack{\ul R\subset R,\ul R \supset S\\ |\ul R|=|R|-1}} \frac{|V^S|}{|V^{R}|} \sum_{\substack{v'\in V^{\ul R}\\ v'[S]=v}} \lambda^{\ul R,R}_{v'}\notag\\
    &=\sum_{\substack{\ul R\subset R,\ul R\supset S\\ |\ul R|=|R|-1}}\overline \lambda^{\ul R,R}_{v}.\notag
    \end{align}
    Adding \eqref{eq:4.1} to \eqref{eq:1.1} for all $R\supset S$ such that $|R|=|S|+2$, we cancel the negative terms of \eqref{eq:1.1} again, obtaining,
    \begin{equation}
        2\left(\hat Y^S_v \sum_{\substack{R\supset S\\ |R|\leq |S|+2}} \frac{1}{\sigma^2(\twid Y^R_{v(R)})} - \sum_{\substack{R\supset S\\ |R|\leq |S|+2}} \frac{\twid Y^R_{v(R)}}{\sigma^2(\twid Y^R_{v(R)})} \right) 
        -\sum_{\substack{\ol R\supset R\supset S\\|\ol R|=|R|+1=|S|+3}}\overline \lambda^{R,\ol R}_v 
        +\sum_{T_i} \lambda^{T_i,S}_{v[T_i]} = 0.
    \end{equation}
    Continuing this process, until we have addressed all $R\supset S$, $R\neq \mscr A$, we obtain:
    \begin{equation}\label{eq:1.2}
         2\left(\hat Y^S_v \sum_{\substack{R\supset S\\ R\neq \mscr A}} \frac{1}{\sigma^2(\twid Y^R_{v(R)})} - \sum_{\substack{R\supset S\\ R\neq \mscr A}} \frac{\twid Y^R_{v(R)}}{\sigma^2(\twid Y^R_{v(R)})} \right) 
        -\sum_{\substack{R\supset  S\\|R| = |\mscr A|-1}}\overline \lambda^{R,\mscr A}_v 
        +\sum_{T_i} \lambda^{T_i,S}_{v[T_i]} = 0.
    \end{equation}
    Finally, we sum up \eqref{eq:3} over all $w\in V^{\mscr A}$ such that $w[S]=v$ and divide by $|V^{\mscr A}|/|V^S|$:
    \begin{equation}\label{eq:5}
        2(\hat Y^S_v - \twid Y^{\mscr A}_{v(\mscr A)})/\sigma^2(\twid Y^{\mscr A}_{v(\mscr A)}) + \sum_{\substack{ R\supset S\\|R|=|\mscr A|-1}} \overline \lambda^{R,\mscr A}_v=0,
    \end{equation}
    where the last term is obtained by simplifying an expression similar to \eqref{eq:simplify}. 

    Adding \eqref{eq:5} to \eqref{eq:1.2}, we have
    \begin{equation}\label{eq:almost}
        2\left(\hat Y^S_v \sum_{R\supset S} \frac{1}{\sigma^2(\twid Y^R_{v(R)})} - \sum_{R\supset S} \frac{\twid Y^R_{v(R)}}{\sigma^2(\twid Y^R_{v(R)})}\right) + \sum_{T_i} \lambda^{T_i,S}_{v[T_i]}=0.
    \end{equation}
    Using the $\acute Y$ notation, we can write this as 
    \begin{equation}
        \label{eq:1.3}
        2(\hat Y^S_v - \acute Y^S_v)/\sigma^2(\acute Y^S_v)+ \sum_{T_i} \lambda^{T_i,S}_{v[T_i]}=0.
    \end{equation}
    Since all of the steps we took to get \eqref{eq:1.3} are invertible, we see that the original solution space is equivalent to the one implied by \eqref{eq:1.3}, \eqref{eq:2}, \eqref{eq:3}, and self-consistency. We see that \eqref{eq:1.3} over all $v\in V^S$ along with the constraint $\hat Y^S_{v(s)} = \hat Y^{T_i}_{v(s)}$ for all $T_i$ and $v\in V^{T_i}$ are equivalent to the following optimization problem:
    \begin{align*}
        &\text{minimize } \sum_{v\in V^S} (\hat Y^S_v - \acute Y^S_v)^2/\sigma^2(\acute Y^S_v)\\
        &\text{subject to } \hat Y^S_{v(S)} = \hat Y^{T_i}_{v},\quad \text{for all $T_i$ and $v\in V^{T_i}$},
    \end{align*}
    which is an equivalent formulation of the down pass solution $\hat Y^S$. Furthermore holding $\hat Y^S$ fixed, \eqref{eq:2}, \eqref{eq:3} and the constraints $\hat Y^{R'}_{v(R')} = \hat Y^R_{v}$ for all $R'\supset R\supset S$ such that $|R'|=|R|+1$ and all $v\in V^R$, are the Lagrange equations for the problem:
    \begin{align*}
        &\text{minimize} \sum_{R\supsetneq S} \sum_{v\in V^R} (\hat Y^R_v - \twid Y^R_v)^2/\sigma^2(\twid Y^R)\\
        &\text{subject to } \hat Y^{\ol R}_{v(\ol R)} = \hat Y^R_v,\quad \text{for all $\ol R\supset R\supset S$ such that $|\ol R|=|R|+1$ and all $v\in V^R$,}\\
        &\phantom{\text{subject to } \hat Y^{\ol R}_{v(\ol R)} = \hat Y^R_v,}\quad\text{where $\hat Y^S$ is held fixed},
    \end{align*}
    which we know has a solution, achieved by a projection. Thus, we have found a solution to the Lagrange equations, which must be equivalent to the corresponding projection. We conclude that under assumptions (A1)-(A3), the SEA BLUE procedure for $\hat Y^S$ is the BLUE for $Y^S$.     
\end{proof}

\lemZ*
\begin{proof}
We only prove the results for $\hat Z^S$ as the argument for $\acute Z^S$ is the same up to a change of notation.  For the reader's convenience, we first recall the definition of $\hat Z_v^S$:

\begin{equation*}
     \hat Z^S_v =\sum_{k=1}^{|S|}\sum_{\substack{\{A_1,\ldots, A_k\}\subset S\\ R = S\setminus \{A_1,\ldots, A_k\}}} (-1)^{k+1} \frac{\hat Y^R_{v[R]}}{|A_1|\cdots |A_k|},
\end{equation*}

    1) Choose $A^*\in S$, set $R=S\setminus\{A^*\}$, and let $v\in V^{R^*}$. Then, using the self-consistency of the prior $\hat Y^R$'s, we can write, 
    \begin{align*}
        \hat Z^S_{v(S)}&= \left(\hat Y^{R}_v+\sum_{\substack{A\in R\\ \ul R=S\setminus \{A,A^*\}}} \frac{\hat Y^{\ul R}_{v[\ul R]}}{|A|}\right)\\
    &-\left(\sum_{\substack{A\in R\\ \ul R = S\setminus \{A,A^*\}}} \frac{\hat Y^{\ul R}_{v[\ul R]}}{|A|} + \sum_{\substack{A_1,A_2\in R\\ \ul R = S\setminus\{A_1,A_2,A^*\}}} \frac{\hat Y^{\ul R}_{v[\ul R]}}{|A_1||A_2|}\right)\\
        &+\left( \sum_{\substack{A_1,A_2\in R\\ \ul R=S\setminus\{A_1,A_2,A^*\}}} \frac{\hat Y^{\ul R}_{v[\ul R]}}{|A_1||A_2|}
        +\sum_{\substack{A_1,A_2,A_3\in R\\
        R = S\setminus \{A_1,A_2,A_3,A^*\}}}\frac{ \hat Y^{\ul R}_{v[\ul R]}}{|A_1||A_2||A_3|}\right)\\
        &+\cdots +
        (-1)^{|S|}\left(\sum_{\substack{A\in R\\ \ul R=\{A\}}} \frac{\hat Y^{\ul R}_{v[\ul R]} }{\prod_{A_i\not\in \{A,A^*\}} |A_i|}
        +\frac{\hat Y^\emptyset }{\prod_{A\neq A^*} |A|}\right)\\
        &+ (-1)^{|S|+1} \frac{\hat Y^\emptyset }{\prod_{A\neq A^*} |A|},\\
    \end{align*}
    which we see is a telescoping sum, canceling all of the terms except the first one, giving the desired result: $Z^*_{v(S)}= \hat Y^{R}_v$. In the above expression, each of the quantities in parentheses corresponds to one of the terms in the outer sum over $k$ of \eqref{eq:Zhat} with the first term consisting of the margins for which $A^*\notin \ul R$ and the second those with $A^*\in \ul R$. Note that if for any particular term, $A^*\in \ul R$, then we marginalize out the $A^*$ variable with no change to the coefficients. On the other hand, if $A^*\not\in \ul R$, then we simply get $|A^*|$ copies of the original term, canceling the factor of $|A^*|$ in the denominator. See Example \ref{ex:Z} for an illustrative example.

    We have established that $\hat Z^S$ is self-consistent with the $\hat Y^R$ for $R\subset S$ which satisfy $|R|=|S|-1$. It follows that $\hat Z^S$ is self-consistent with the higher margins, as they are already self-consistent with each other. 

    2) It suffices to show that  for all $Z^0\in \mscr S_S$, $\sum_{v\in V^S} Z^0_v \hat Z^S_v=0$:
    \begin{align*}
       \sum_{v\in V^S} Z_v^0\hat Z^S_v
       &=\sum_{v\in V^S}Z_v^0\left[ \sum_{k=1}^{|S|}\sum_{\substack{\{A_1,\ldots, A_k\}\subset S\\ \ul S=S\setminus \{A_1,\ldots, A_k\}}} \frac{(-1)^{k+1} \hat Y^{\ul S}_{v[\ul S]}}{\prod_{i=1}^k |A_i|}\right]\\
       &=\sum_{k=1}^{|S|}\sum_{\substack{\{A_1,\ldots, A_k\}\subset S\\ \ul S=S\setminus \{A_1,\ldots, A_k\}}}\sum_{v\in V^S} \frac{(-1)^{k+1}Z_v^0\hat Y^{\ul S}_{v[\ul S]}}{\prod_{i=1}^k |A_i|}\\ 
       &=\sum_{k=1}^{|S|}\sum_{\substack{\{A_1,\ldots, A_k\}\subset S\\ \ul S=S\setminus \{A_1,\ldots, A_k\}}}\sum_{w\in V^{\ul S}}\sum_{\substack{v\in V^S\\w=v[\ul S]}} \frac{(-1)^{k+1}Z_v^0\hat Y^{\ul S}_{w}}{\prod_{i=1}^k |A_i|}\\ 
       &=\sum_{k=1}^{|S|}\sum_{\substack{\{A_1,\ldots, A_k\}\subset S\\ \ul S=S\setminus \{A_1,\ldots, A_k\}}}\sum_{w\in V^{\ul S}} \frac{(-1)^{k+1}\hat Y^{\ul S}_{w}}{\prod_{i=1}^k |A_i|}\sum_{\substack{v\in V^S\\w=v[\ul S]}}Z_v^0.\\ 
    \end{align*}
   Since $Z^0\in \mscr S_S$, we have that $\sum_{\substack{v\in V^S\\w=v[\ul S]}}Z_v^0=Z^0_{w(S)}=0$. We conclude that $\sum_{v\in V^S} Z_v^0\hat Z^S_v=0$ as desired.
\end{proof}

\thmFastProj*
\begin{proof}
    Call $P$ the operator that maps $\acute Y^S$ to $\acute Z^S$. To simplify notation, we write $\mscr S$ in place of $\mscr S_S$.
    
    1) We will show that $P = \mathrm{Proj}^I_{\mscr S^\perp}$. It suffices to show that a) the elements of $\mscr S$ are eigenvectors of $P$ with eigenvalue 0, b) the elements of $\mscr S^\perp$ are eigenvectors of $P$ with eigenvalue 1, and c) the operator $P$ is self-adjoint with respect to the standard inner product. This suffices because a) and b) establish that $P$ is idempotent, with kernel $\mscr S$ and range $\mscr S^\perp$, and c) establishes that $P$ is symmetric. Altogether these imply that $P$ is the orthogonal projection onto $\mscr S^\perp$. 

    a) Suppose $\acute Y^S\in \mscr S$, meaning that all margins of $\acute Y^S$ are zero. By the construction of $P$ in \eqref{eq:Zacute}, we have that $\acute Z^S_v=0$ for all $v\in S$. Thus, elements of $\mscr S$ are eigenvectors of $P$ with eigenvalue $0$.

    b) Suppose that $\acute Y^S\in \mscr S^\perp$. We know that $\acute Z^S\in \mscr S^\perp$ by Lemma \ref{lem:Z}. Since $\mscr S^\perp$ is closed under linear combinations, $\acute Z^S-\acute Y^S\in \mscr S^\perp$. However, since $\acute Z^S$ and $\acute Y^S$ have the same margins, by Lemma \ref{lem:Z}, we have that $\acute Z^S-\acute Y^S\in \mscr S$. Thus, $\acute Z^S - \acute Y^S \in \mscr S\cap \mscr S^\perp = \{0\}$, since $\mscr S$ and $\mscr S^\perp$ are orthogonal spaces. Thus, $\acute Z^S=\acute Y^S$. We see that the elements of $\mscr S^\perp$ are eigenvectors of $P$ with eigenvalue $1$.

    c) Let $X$ and $Y$ be two tables of the same dimension. We need to show that $(PX)^\top Y = X^\top (PY)$: 

    \begin{align*}
        (PX)^\top Y &= (PX)^\top (Y - PY + PY)\\
        &=(PX)^\top (Y-PY) + (PX)^\top (PY)\\
        &=(PX)^\top (PY),
    \end{align*}
    where we used the fact that $PX\in \mscr S^\perp$, and $Y-PY\in \mscr S$ which can be seen as follows: let $V_1,\ldots, V_{\dim \mscr S}$ be a basis for $\mscr S$ and let $W_1,\ldots, W_{\dim \mscr S^\perp}$ be a basis for $\mscr S^\perp$, which together form a basis for all tables on the variables $S$. Then there exists real-valued coefficients $c_1,\ldots, c_{\dim \mscr S}$ and $d_1,\ldots, d_{\dim \mscr S^\perp}$ such that 
    \[Y = \sum_{i=1}^{\dim \mscr S} c_i V_i + \sum_{j=1}^{\dim \mscr S^\perp} d_j W_j.\]
    Then,
    \[PY = \sum_{j=1}^{\dim \mscr S^\perp}  d_j W_j,\]
    and we see that $Y-PY = \sum_{i=1}^{\dim \mscr S} c_i V_i \in \mscr S$. 

    By swapping the roles of $X$ and $Y$ above, we can similarly show that $X^\top (PY) = (PX)^\top (PY)$, which implies that $(PX)^\top Y = X^\top (PY)$, establishing that $P$ is self-adjoint in the standard inner product.

    We have established that $P$ is the orthogonal projection with kernel $\mscr S$ and range $\mscr S^\perp$. 

    2) By Theorem \ref{thm:general}, we know that $\hat Y^S = \mathrm{Proj}^{\acute \Sigma^{-1}}_{\mscr S} \acute Y^S + \mathrm{Proj}_{\mscr S^\perp}^{\acute \Sigma} \hat Z^S$ since $\hat Z^S$ is a table with the desired margins, where $\acute \Sigma$ is the covariance of $\acute Y^S$. Under assumptions (A1)-(A3), we have that all elements of $\acute \Sigma$ are independent with equal variance. It follows that the projections are orthogonal projections. So, we can simplify this as,
    \begin{align*}
        \hat Y^S &= \acute Y^S-\acute Z^S+P\hat Z^S\\
       &= \acute Y^S-\acute Z^S+\hat Z^S,
    \end{align*}
    since $\acute Y^S-\acute Z^S = \mathrm{Proj}^I_{\mscr S}$, $P=\mathrm{Proj}^I_{\mscr S^\perp}$, and $\hat Z^S\in \mscr S^\perp$.
\end{proof}


The lemma below quantifies the dimension of $\mscr S_S$, which is defined in Equation \eqref{eq:S}. 
\begin{lem}\label{lem:basis}
    Let $S\in \mscr D$ consist of $k$ variables, each with $n_j$ levels. Let $\mscr S_S$ be defined as in \eqref{eq:S}.
    \begin{enumerate}
        \item The dimension of $\mscr S_S$ is
    \[\sum_{j=0}^{k}(-1)^{k-j}\sum_{\{m_1,\ldots, m_{j}\}\subset \{n_1,\ldots, n_k\}}m_1\cdots m_j=\prod_{i=1}^k (n_i-1).\] 
        \item A basis for $\mscr S_S$ is the set of tables $Z^{(i_j)_{j=1}^k}$ for $i_j\in \{1,\ldots, n_j-1\}$, where for $v\in V^S$, 
    \[Z_v^{(i_j)_{j=1}^k} = \left[ \prod_{j=1}^k (-1)^{I(v_j=i_j+1)} I(v_j\in \{i_j,i_j+1\})\right].\]
    \end{enumerate}
    
\end{lem}
\begin{proof}
 Recall that $\mscr S_S=\{Z^S|\text{for all } \ul S\subsetneq S, \text{ and all }v\in V^{\ul S}, \text{ we have }Z^S_{v(S)}=0\}$. To establish the dimension of $\mscr S_S$, we will instead calculate the dimension of its orthogonal complement $\mscr S_S^\perp$. We know that $\mscr S_S^\perp$ is spanned by the set of self-consistency constraints, which enforce  $Z_{v(S)}=0$ for $\ul S\subset S$ such that $|\ul S|=|S|-1$ and all $v\in V^{\ul S}$ (note that requiring every margin to be equal to 0 is equivalent to only requiring that the $|S|-1$ margins are equal to 0). We observe that there are $\sum_{\{m_1,\ldots, m_{k-1}\}\subset \{n_1,\ldots, n_k\}} m_1m_2\cdots m_{k-1}$ such constraints, but they are not linearly independent. Indeed, for any such $\ul S$, the set of constraints $\{Z_{v(S)}=0$ for $v\in V^{\ul S}\}$ also implies that higher margins are equal to zero, giving redundant information for different $\ul S$'s. Using inclusion-exclusion, we can count the rank of these constraints as 
    \[\sum_{j=0}^{k-1}(-1)^{k-j-1}\sum_{\{m_1,\ldots, m_{j}\}\subset \{n_1,\ldots, n_k\}}m_1\cdots m_j.\] 
    Then, the dimension of $\mscr S_S$ is 
    \begin{align*}
        &(n_1\cdots n_k)-\sum_{j=0}^{k-1}(-1)^{k-j-1}\sum_{\{m_1,\ldots, m_{j}\}\subset \{n_1,\ldots, n_k\}}m_1\cdots m_j\\
        &=\sum_{j=0}^{k}(-1)^{k-j}\sum_{\{m_1,\ldots, m_{j}\}\subset \{n_1,\ldots, n_k\}}m_1\cdots m_j\\
        &=\prod_{i=1}^k (n_i-1).
    \end{align*}

    For part 2, note that each $Z_v^{(i_j)_{j=1}^k}$ is a member of $\mscr S_S$ since each table has a single ``$+1$'' and ``$-1$'' in each coordinate. Furthermore, the $Z_v^{(i_j)_{j=1}^k}$ can be seen to be linearly independent since in lexicographic order, they form a lower triangular matrix. Finally, we count that there are $\prod_{i=1}^k(n_i-1)$ such tables, which is equal to the rank of $\mscr S_S$, shown in part 1. We conclude that the $Z_v^{(i_j)_{j=1}^k}$ indeed form a basis for $\mscr S_S$.  
\end{proof}

\lemP*
\begin{proof}
    In Lemma \ref{lem:basis}, we established that the dimension of the null space when making a table with $i$ variables self-consistent with the margins one level above is $(I-1)^i$. These constraints can be seen to be linearly independent, resulting in 
    \[\sum_{i=0}^k \binom ki (I-1)^i = I^k,\]
    as the dimension of the complete null space. Thus, 
    $p = \min\{I^k,(I+1)^k-I^k\}.$ 
    The runtime and memory results follow by using the formulas $O(np+p^{2.373})$ for runtime and $O(np)$ for memory, along with the simplification $p = O(I^{k-1})$ as $I\rightarrow \infty$ and $p=O(I^k)$ as $k\rightarrow \infty$. 
\end{proof}

\begin{restatable}{lem}{lem:Ztime}\label{lem:Ztime}
    Let $S$ be a set of variables, with size $|S|=k$. Given self-consistent margins $\hat Y^R$ for $R\subsetneq S$, the time to produce $\hat Z^S$ is $O(k2^{k-1})$.
\end{restatable}
\begin{proof}
The construction in \eqref{eq:Zhat} requires summing over all proper subsets of $S$, and calculating $\prod_{i=1}^j |A_i|$ for a subset $\{A_1,\ldots, A_j\}\subset S$. The product takes $O(j)$ time, so we have a runtime on the order of 
\begin{align*}
    \sum_{j=1}^{k} \binom{k}{j} k = k2^{k-1},
\end{align*}
since there are $\binom{k}j$ subsets of $S$ of size $j$. 
\end{proof}

 \lemCollectionCost*
\begin{proof}
1) We will start by marginalizing out one variable at a time and using previously computed margins as the starting point for subsequent margins to avoid marginalizing the same variable more than once. Suppose we consider all tables which have marginalized over $i$ variables. There are $\binom ki$ such tables, which consist of $I^{k-i}$ counts each; for each count in a margin, it takes $I$ time to compute it from a previous margin. Thus, the runtime is 
\begin{align*}
    \sum_{i=1}^k \binom ki I^{k-i+1} &=I(I+1)^k-I^{k+1},
\end{align*}
which is $O(I^k)$ as $I\rightarrow \infty$ and $O((I+1)^k)$ as $k\rightarrow \infty$. The memory required (including the original counts in $S$) is 
\begin{align*}
    \sum_{i=0}^k \binom ki I^{k-i}=(I+1)^k,
\end{align*}
which is $O(I^k)$ as $I\rightarrow \infty$ and $O((I+1)^k)$ as $k\rightarrow \infty$. 

2) Using the above result for a single table and summing over all tables, we have a total runtime of 
\begin{align*}
    \sum_{j=0}^k \binom kj \sum_{i=1}^j \binom ji I^{k-i+1}
    &=I(2I+1)^k-2^kI^{k+1},\\
\end{align*}
which is $O((I+1)^k)=O(n)$ as $I\rightarrow \infty$ , and is $O((2I+1)^k)=O(n^{\log(2I+1)/\log(I+1)})$ as $k\rightarrow \infty$. 

In the collection step, every table produces at most $O(n)$ estimates, which are aggregated into the running totals $C_v^S$ and $D_v^S$, which are of order $O(n)$. Thus, at no point does the memory requirement exceeed $O(n)$. For the down pass, when (A1)-(A3) hold, we can use the efficient formula for $\hat Y^S$ in Theorem \ref{thm:FastProj} which, as Remark \ref{rem:downpass} points out, has the same computational and memory requirement as the collection step. Note that Lemma \ref{lem:Ztime} establishes that the runtime to assemble $\acute Z^S$ and $\hat Z^S$, given the margins, is of a lower order than the time to compute the margins in the first place. 
\end{proof}

\propcovarianceExact*
\begin{proof}
    For simplicity, we write $P=\mathrm{Proj}_{S_0}^{\Sigma^{-1}}$, where $P$ is the projection matrix defined in Theorem \ref{thm:general} which is implemented by SEA BLUE. We can write 
    \begin{align*}
        P\Sigma P^\top &= P\Sigma^{1/2} \Sigma^{1/2} P^\top\\
        &=P \Sigma^{1/2}(P\Sigma^{1/2})^\top,
    \end{align*}
    where we use the fact that $\Sigma^{1/2}$ is symmetric. In the case that $\Sigma=\sigma^2 I$, we have that $P$ is an orthogonal projection, so $P^\top=P$ and by the idempotency of projection matrices we have 
    \[P\sigma^2 IP^\top=\sigma^2 PP=\sigma^2 P.\]
    The correctness of the confidence interval follows from Theorem \ref{thm:general}.
\end{proof}

\propMonteCarloT*
\begin{proof}
    Even when (A2) and (A3) may not hold, SEA BLUE is still a linear procedure. It is well known that a linear transformation of a multivariate normal distribution is another multivariate normal. Since the mean of $\hat Y^{(j)}$ is known to be zero, $(\sigma^2)^m_i$ is an unbiased estimator for the variance of $\vect(\hat Y)_i$ and we have that $\frac{R (\sigma^2)^R_i}{\var(\vect(\hat Y)_i)}\sim \chi^2(R)$. Then,
    \[\frac{\vect(\hat Y)_i-\vect(Y)_i}{\sqrt{(\sigma^2)^R_i}}\sim t(R),\]
    and the validity of the confidence interval follows.
\end{proof}

 \propMonteCarloDF*
 \begin{proof}
 Note that $|\vect(\hat Y)_i - \vect(Y)_i|$, $|\vect(\hat Y^{(1)}_i)|,\ldots, |\vect(\hat Y^{(m)}_i)|$ are all i.i.d., implying that any ordering is equally likely. So, 
     \begin{align*}
         P\left(|\vect(\hat Y)_i - \vect(Y)_i|\leq  X^i_{(\lceil (1-\alpha)(R+1)\rceil)}\right)
         &\geq \frac{\lceil (1-\alpha)(R+1)\rceil}{R+1}
         \geq \frac{(1-\alpha)(R+1)}{R+1}
         =1-\alpha,
     \end{align*}
     where we used the fact that SEA BLUE is a linear and unbiased procedure to establish that $|\vect(\hat Y_i)-\vect(Y_i)|$ is equal in distribution to the $|\vect(\hat Y_i^{(j)})|$'s, which is used for the first inequality.  Note that the first inequality is actually an equality if the variables are continuous, and the inequality may be strict when ties are possible. 
     Thus, the interval $\vect(\hat Y)_i \pm X^i_{(\lceil (1-\alpha) (R+1)\rceil)}$ 
     is a $(1-\alpha)$-confidence interval for $\vect(Y)_i$.
 \end{proof}

\bibliography{references}

\end{document}